\documentclass[11pt]{article}
\usepackage[margin=1.5in]{geometry}
\usepackage[utf8]{inputenc}

\usepackage[pdftex]{graphicx}

\usepackage{pdfsync}

\usepackage{microtype}
\usepackage{wrapfig}
\usepackage{amsmath, amsfonts, amssymb, amsthm}
\usepackage{MnSymbol}
\usepackage{mathtools}
\usepackage{balance}
\usepackage{enumerate}
\usepackage{breakcites}
\usepackage{url}

\usepackage{tikz}
\usetikzlibrary{decorations.pathreplacing}
\usetikzlibrary{math}
\usepackage{thmtools,thm-restate}

\newtheorem{definition}{Definition}
\newtheorem{theorem}{Theorem}
\newtheorem*{theorem*}{Theorem}

\newtheorem{corollary}{Corollary}
\newtheorem{lemma}{Lemma}

\DeclareMathOperator*{\HE}{H}
\DeclareMathOperator*{\E}{E}

\DeclareMathOperator*{\dist}{dist}
\DeclareMathOperator*{\poly}{poly}
\DeclareMathOperator*{\simi}{sim}

\newcommand{\zo}{\{0,1\}}

\title{
   {Optimal Las Vegas Locality Sensitive Data Structures}\\
   {\large Full Version}
}
\author{
   Thomas Dybdahl Ahle\\
   {IT University of Copenhagen}
}
\date{June 27 2018}

\begin{document}
\maketitle

%IEEE says to remove math from the abstract.
% That would make things quite hard...
\begin{abstract}
   We show that approximate similarity (near neighbour) search can be solved in high dimensions with performance matching state of the art (data independent) Locality Sensitive Hashing, but with a guarantee of no false negatives.
   Specifically we give two data structures for common problems.

   For $c$-approximate near neighbour in Hamming space we get query time $dn^{1/c+o(1)}$ and space $dn^{1+1/c+o(1)}$ matching that of \cite{indyk1998approximate} and answering a long standing open question from~\cite{indyk2000dimensionality} and~\cite{pagh2016locality} in the affirmative.
   By means of a new deterministic reduction from $\ell_1$ to Hamming we also solve $\ell_1$ and $\ell_2$ with query time $d^2n^{1/c+o(1)}$ and space $d^2 n^{1+1/c+o(1)}$.

   For $(s_1,s_2)$-approximate Jaccard similarity we get query time $dn^{\rho+o(1)}$ and space $dn^{1+\rho+o(1)}$, $\rho=\log\frac{1+s_1}{2s_1}\big/\log\frac{1+s_2}{2s_2}$, when sets have equal size, matching the performance of~\cite{tobias2016}.

   The algorithms are based on space partitions, as with classic LSH, but we construct these using a combination of brute force, tensoring, perfect hashing and splitter functions à la~\cite{naor1995splitters}.
   We also show a new dimensionality reduction lemma with 1-sided error.
\end{abstract}

\section{Introduction}
Locality Sensitive Hashing has been a leading approach to high dimensional similarity search (nearest neighbour search) data structures for the last twenty years.
Intense research
\cite{indyk1998approximate, gionis1999similarity, kushilevitz2000efficient, indyk2000high, indyk2001approximate, charikar2002similarity, datar2004locality, lv2007multi, panigrahy2006entropy, andoni2006near, andoni2014beyond, andoni2016optimal, becker2016new, ahle2017parameter, aumuller2017distance} has applied the concept of space partitioning to many different problems and similarity spaces.
These data structures are popular in particular because of their ability to overcome the `curse of dimensionality' and conditional lower bounds by~\cite{williams2005new}, and give sub-linear query time on worst case instances.
They achieve this by being approximate and Monte Carlo, meaning they may return a point that is slightly further away than the nearest, and with a small probability they may completely fail to return any nearby point.

\begin{definition}[$(c,r)$-Approximate Near Neighbour]
   Given a set $P$ of $n$ data points in a metric space $(X, \dist)$,
   build a data structure, such that given any $q\in X$, for which there is an $x\in P$ with $\dist(q,x)\le r$, we return a $x'\in P$ with $\dist(q,x')\le cr$.
\end{definition}

A classic problem in high dimensional geometry has been whether data structures existed for $(c,r)$-Approximate Near Neighbour with Las Vegas guarantees, and performance matching that of Locality Sensitive Hashing.
That is, whether we could guarantee that a query will always return an approximate near neighbour, if a near neighbour exists; or simply, if we could rule out false negatives?
The problem has seen practical importance as well as theoretical.
There is in general no way of verifying that an LSH algorithm is correct when it says `no near neighbours' - other than iterating over every point in the set, in which case the data structure is entirely pointless.
This means LSH algorithms can't be used for many critical applications, such as finger print data bases.
Even more applied, it has been observed that tuning the error probability parameter is hard to do well, when implementing LSH~\cite{gionis1999similarity, arya1998optimal}.
A Las Vegas data structure entirely removes this problem.
Different authors have described the problem with different names, such as `Las Vegas'~\cite{indyk2000dimensionality}, `Have no false negatives'~\cite{goswami2017distance, pagh2016locality}, `Have total recall'~\cite{pham2016scalability}, `Are exact'~\cite{arasu2006efficient} and `Are explicit'~\cite{karppa2016explicit}.

Recent years have shown serious progress towards finally solving the problem.
In particular~\cite{pagh2016locality} showed that the problem in Hamming space admits a Las Vegas algorithm with query time $dn^{1.38/c+o(1)}$, matching the $dn^{1/c}$ data structure of~\cite{indyk1998approximate} up to a constant factor in the exponent.
In this paper we give an algorithm in the Locality Sensitive Filter framework~\cite{becker2016new, christiani2016framework}, which not only removes the factor $1.38$, but improves to $dn^{1/(2c-1)+o(1)}$ in the case $cr\approx d/2$, matching the algorithms of~\cite{andoni2015practical} for Hamming space.

We would like to find an approach to Las Vegas LSH that generalizes to the many different situations where LSH is useful.
Towards that goal, we present as second algorithm for the approximate similarity search problem under Braun-Blanquet similarity, which is defined for sets $x,y\subseteq[d]$ as $\simi(x,y)=|x\cap y|/\max(|x|,|y|)$.
We refer to the following problem definition:

\begin{definition}[Approximate similarity search]
   Let $P \subseteq \mathcal P([d])$ be a set of $|P| = n$ subsets of $[d]$; (here $\mathcal P(X)$ denotes the powerset of $X$.)
   let $\simi : \mathcal P([d]) \times \mathcal P([d]) \to [0,1]$ be a similarity measure.
   For given $s_1, s_2 \in [0,1]$, $s_1 > s_2$,
   a solution to the ``$(s_1, s_2)$-similarity search problem under $\simi$'' is a data structure that supports the following query operation:
   on input $q \subseteq [d]$, for which there exists a set $x \in P$ with $\simi(x,q) \ge s_1$, return $x' \in P$ with $\simi(x', q) > s_2$.
\end{definition}

The problem has traditionally been solved using the Min-Hash LSH~\cite{broder1997syntactic, broder1997resemblance}, which combined with the results of Indyk and Motwani~\cite{indyk1998approximate} gives a data structure with query time $dn^\rho$ and space $dn^{1+\rho}$ for $\rho=\log s_1/\log s_2$.
Recently it was shown by~\cite{tobias2016} that this could be improved for vectors of equal weight to $\rho=\log\frac{2 s_1}{1+s_1}\big/\log\frac{2 s_2}{1+s_2}$.
We show that it is possible to achieve this recent result with a data structure that has no false negatives.

\subsection{Summary of Contributions}

We present the first Las Vegas algorithm for approximate near neighbour search, which gives sub-linear query time for any approximation factor $c > 1$.
This solves a long standing open question from~\cite{indyk2000dimensionality} and~\cite{pagh2016locality}.
In particular we get the following two theorems:

\begin{theorem}\label{thm:ham}
   Let $X=\zo^d$ be the Hamming space with metric $\dist(x,y) = \|x \oplus y\| \in [0,d]$ where $\oplus$ is ``xor'' or addition in $\mathbb Z_2$.
   For every choice of $0 < r$, $1 < c$ and $cr \le d/2$,
   we can solve the $(c, r)$-approximate near neighbour problem in Hamming space
   with query time $dn^\rho$ and space usage $dn+n^{1+\rho}$ where $\rho = 1/c + \hat O((\log n)^{-1/4})$.
\end{theorem}

Note: $\hat O$ hides $\log\log n$ factors.

\begin{corollary}\label{cor:ham}
   When $r/d=\Omega((\log n)^{-1/6})$, we get the improved exponent $\rho = \frac{1-cr/d}{c(1-r/d)} + \hat O((\log n)^{-1/3}d/r)$.
\end{corollary}

This improves upon theorem~\ref{thm:ham} when $r/d$ is constant (or slightly sub-constant), including in the important ``random case'', when $r/d = 1/(2c)$ where we get $\rho=1/(2c-1)+o(1)$.

\begin{theorem}\label{thm:sim}
   Let $\simi$ be the Braun-Blanquet similarity $\simi(x,y) = |x\cap y|/\max(|x|,|y|)$.
   For every choice of constants $0 < s_2 < s_1 < 1$, we can solve the $(s_1, s_2)$-similarity problem over $\simi$ with query time
   $dn^\rho$ and space usage $dn+ n^{1+\rho}$ where $\rho = \log s_1/\log s_2 + \hat O((\log n)^{-1/2})$.
\end{theorem}

For sets of fixed size $w$, the $dn$ terms above can be improved to $wn$.
It is also possible to let $s_1$ and $s_2$ depend on $n$ with some more work.

The first result matches the lower bounds by
%~ Removed to fill hbox
\cite{o2014optimal} for ``data independent'' LSH data structures for Hamming distance and improves upon~\cite{pagh2016locality} by a factor of $\log 4>1.38$ in the exponent.
By deterministic reductions from $\ell_2$ to $\ell_1$~\cite{indyk2007uncertainty} and $\ell_1$ to hamming (appendix \ref{ell1tohamming}), this also gives the best currently known Las Vegas data structures for $\ell_1$ and $\ell_2$ in $\mathbb R^d$.
The second result matches the corresponding lower bounds by~\cite{tobias2016} for Braun-Blanquet similarity and, by reduction, Jaccard similarity.
See table~\ref{tab:results} for more comparisons.

Detaching the data structures from our constructions, we give the first explicit constructions of large Turán Systems~\cite{sidorenko1995we}, which are families $\mathcal T$ of $k$-subsets of $[n]$, such that any $r$-subset of $[n]$ is contained in a set in $\mathcal T$.
Lemma~\ref{lem:turan} constructs $(n,k,r)$-Turán Systems using $(n/k)^r e^\chi$ sets, where $\chi=O(\sqrt{r}\log r+\log k+\log\log n)$.
For small values of $k$ this is sharp with the lower bound of $\binom{n}{r}/\binom{k}{r}$, and our systems can be efficiently decoded, which is likely to have other algorithmic applications.

\subsection{Background and Related Work}
The arguably most successful technique for similarity search in high dimensions is Locality-Sensitive Hashing (LSH), introduced in 1998 by~\cite{indyk1998approximate, har2012approximate}.
The idea is to make a random space partition in which similar points are likely to be stored in the same region, thus allowing the search space to be pruned substantially.
The granularity of the space partition (the size/number of regions) is chosen to balance the expected number of points searched against keeping a (reasonably) small probability of pruning away the actual nearest point.
To ensure a high probability of success (good recall) one repeats the above construction, independently at random, a small polynomial (in $n$) number of times.

In \cite{pagh2016locality, arasu2006efficient} it was shown that one could change the above algorithm to not do the repetitions independently.
(Eliminating the error probability of an algorithm by independent repetitions, of course, takes an infinite number of repetitions.)
By making correlated repetitions, it was shown possible to reach zero false negatives much faster, after only polynomially many repetitions.
This means, for example, that they needed more repetitions than LSH does to get $0.99$ success rate, but fewer than LSH needs for success rate $1-2^{-n}$.

An alternative to LSH was introduced by~\cite{becker2016new, dubiner2010bucketing}.
% TODO: Rasmus suggests mentioning Tobias' new paper here, as having developed LSF further.
It is referred to as Locality Sensitive Filters, or LSF.
While it achieves the same bounds as LSH, LSF has the advantage of giving more control to the algorithm designer for balancing different performance metrics.
For example, it typically allows better results for low dimensional data, $d=O(\log n)$, and space/time trade-offs~\cite{andoni2016optimal}.
The idea is to sample a large number of random sections of the space.
In contrast to LSH these sections are not necessarily partitions and may overlap heavily.
For example, for points on the sphere $S^{d-1}$ the sections may be defined by balls around the points of a spherical code.
One issue compared to LSH is that the number of sections in LSF is very large.
This means we need to impose some structure so we can efficiently find all sections containing a particular point.
With LSH the space partitioning automatically provided such an algorithm, but for LSF it is common to use a kind of random product code.
(An interesting alternative is~\cite{tobias2016}, which uses a random branching processes.)
LSF is similar to LSH in that it only approaches 100\% success rate as the number of sections goes to infinity.

The work in this paper can be viewed as way of constructing correlated, efficiently decodable filters for Hamming space and Braun-Blanquet similarity.
That is, our filters guarantee that any two close points are contained in a shared section, without having an infinite number of sections.
Indeed the number of sections needed is equal to that needed by random constructions for achieving constant success probability, up to $n^{o(1)}$ factors.
It is not crucial that our algorithms are in the LSF framework rather than LSH.
Our techniques can make correlated LSH space partitions of optimal size as well as filters.
However the more general LSF framework allows for us to better show of the strength of the techniques.

One very important line of LSH/LSF research, that we don't touch upon in this paper, is that of data dependency.
In the seminal papers \cite{andoni2014beyond, andoni2015optimal, andoni2016optimal} it was shown that the performance of space partition based data structures can be improved, even in the worst case, by considering the layout of the points in the data base.
Using clustering, certain bad cases for LSH/LSF can be removed, leaving only the case of ``near random'' points to be considered, on which LSH works very well.
It seems possible to make Las Vegas versions of these algorithms as well, since our approach gives the optimal performance in these near random cases.
However one would need to find a way to derandomize the randomized clustering step used in their approach.

There is of course also a literature of deterministic and Las Vegas data structures not using LSH.
As a baseline, we note that the ``brute force'' algorithm that stores every data point in a hash table, and given a query, $q\in\zo^d$, looks up every $\sum_{k=1}^r{d\choose k}$ point of Hamming distance most $r$.
This requires $r\log(d/r) < \log n$ to be sub-linear, so for a typical example of $d=(\log n)^2$ and $r=d/10$ it won't be practical.
In \cite{cole2004dictionary} this was somewhat improved to yield $n (\log n)^r$ time, but it still requires $r=O(\frac{\log n}{\log\log n})$ for queries to be sub-linear.
We can also imagine storing the nearest neighbour for every point in $\{0,1\}^d$.
Such an approach would give fast (constant time) queries, but the space required would be exponential in $r$.

In Euclidean space ($\ell_2$ metric) the classical K-d tree algorithm~\cite{bentley1975multidimensional} is of course deterministic, but it has query time $n^{1-1/d}$,
so we need $d=O(1)$ for it to be strongly sub-linear.
Allowing approximation, but still deterministic, \cite{arya1998optimal} found a $(\frac d{c-1})^d$ algorithm for $c>1$ approximation.
They thus get sublinear queries for $d=O(\frac{\log n}{\log\log n})$.

For large approximation factors \cite{har2012approximate} gave a deterministic data structure with query time $O(d\log n)$, but space and preprocessing more than $n\cdot O(1/(c-1))^d$.
In a different line of work, \cite{indyk2000dimensionality} gave a deterministic $(d\epsilon^{-1}\log n)^{O(1)}$ query time, fully deterministic algorithm with space usage $n^{O(1/\epsilon^6)}$ for a $3+\epsilon$ approximation.

See Table~\ref{tab:results} for an easier comparison of the different results and spaces.

\def\arraystretch{1.5}
\begin{table*}
   \center
   \begin{tabular}{| p{3cm} | p{1.9cm} | p{2cm} | p{6.9cm} |}
     \hline
     \bfseries Reference
     & \bfseries Space
     & \bfseries Exponent, search time
     & \bfseries Comments
     \\\hline
     \cite{bentley1975multidimensional}
     & $\ell_2$
     & $1-1/d$
     & Exact algorithm, Fully deterministic.
     \\
     \cite{cole2004dictionary}
     & Hamming
     & $r \frac{\log\log n}{\log n}$
     & Sub-linear for $r < \frac{\log n}{\log\log n}$. Exact.
     \\
     \cite{arya1998optimal}
     & $\ell_2$
     & $d\frac{\log(d/(c-1))}{\log n}$
     & Sub-linear for $d < \frac{\log n}{\log\log n}$.
     \\
     \hline
     \cite{har2012approximate}
     & Hamming
     & $o(1)$
     & $c$-approximation, Fully deterministic, $(1/(c-1))^d$ space.
     \\
     \cite{indyk2000dimensionality}
     & Hamming
     & $o(1)$
     & $(3+\epsilon)$-approximation, Fully deterministic, $n^{\Omega(1/\epsilon^6)}$ space.
     \\
     \cite{arasu2006efficient}
     & Hamming
     & $\approx3/c$
     & The paper makes no theoretical claims on the exponent.
     \\
     \cite{pagh2016locality}
     & Hamming
     & $1.38/c$
     & Exponent $1/c$ when $r=o(\log n)$ or $(\log n)/(cr)\in \mathbb N$.
     \\
     \cite{pacuk2016locality}
     & $\ell_p$
     & $O(d^{1-1/p}/c)$
     & Sub-linear for $\ell_2$ when $c=\omega(\sqrt{d})$.
     \\
     \textbf{This paper}
     & Hamming, $\ell_1$, $\ell_2$
     & $1/c$
     & Actual exponent is $\frac{1-cr/d}{c(1-r/d)}$ which improves to $1/(2c-1)$ for $cr \approx d/2$.
     \\
     \hline
     \cite{pagh2016locality}
     & Braun-Blanquet
     & $1.38\,\frac{1-b_1}{1-b_2}$
     & Via reduction to Hamming. Requires sets of equal weight.
     \\
     \textbf{This paper}
     & Braun-Blanquet
     & $\frac{\log1/b_1}{\log1/b_2}$
     & See~\cite{tobias2016} figure 2 for a comparison with~\cite{pagh2016locality}.
     \\\hline
  \end{tabular}
  \caption{
     Comparison of Las Vegas algorithms for high dimensional near neighbour problems.
     The exponent is the value $\rho$, such that the data structure has query time $n^{\rho + o(1)}$.
     All listed algorithms, except for~\cite{indyk2000dimensionality} use less than $n^2$ space.
     All algorithms give $c$-approximations, except for the first two, and for ~\cite{indyk2000dimensionality}, which is a $(3+\epsilon)$-approximation.
  }
   \label{tab:results}
\end{table*}

\subsection{Techniques}

Our main new technique is a combination of `splitters' as defined by ~\cite{naor1995splitters, alon2006algorithmic}, and `tensoring' which is a common technique in the LSH literature.

Tensoring means constructing a large space partition $P\subseteq \mathcal{P}(X)$ by taking multiple smaller random partitions $P_1, P_2, \dots$ and taking all the intersections $P = \{p_1\cap p_2,\, \dots \mid p_1\in P_1,\, p_2\in P_2, \dots\}$.
Often the implicit partition $P$ is nearly as good as a fully random partition of equal size, while it is cheaper to store in memory and allows much faster lookups of which section covers a given point.
In this paper we are particularly interested in $P_i$'s that partition different small sub-spaces, such that $P$ is used to increase the dimension of a small, explicit, good partition.

Unfortunately tensoring doesn't seem to be directly applicable for deterministic constructions, since deterministic space partitions tend to have some overhead that gets amplified by the product construction.
This is the reason why \cite{pagh2016locality} constructs hash functions directly using algebraic methods, rather than starting with a small hash function and `amplifying' as is common for LSH.
Algebraic methods are great when they exist, but they tend to be hard to find, and it would be a tough order to find them for every similarity measure we would like to make a data structure for.

It turns out we can use splitters to help make tensoring work deterministically.
Roughly, these are generalizations of perfect hash functions.
However, where a $(d,m,k)$-perfect hash family guarantees that for any set $S\subseteq[d]$ of size $k$, there is a function $\pi : [d] \to [m]$ such that $|\pi(S)|=k$, a $(d,m)$-splitter instead guarantees that the is some $\pi$ such that $|S\cap\pi^{-1}(i)|=d/m$ for each $i=1,\dots,m$; or as close as possible if $m$ does not divide $d$.
That is, for any $S$ there is some $\pi$ that `splits' $S$ evenly between $m$ buckets.

Using splitters with tensoring, we greatly limit the number of combinations of smaller space partitions that are needed to guarantee covering.
We use this to amplify partitions found probabilistically and verified deterministically.
The random aspect is however only for convenience, since the greedy set cover algorithm would suffice as well, as is done in~\cite{alon2006algorithmic}.
We don't quite get a general reduction from Monte Carlo to Las Vegas LSH data structures, but we show how two state of the art algorithms may be converted at a negligible overhead.

\smallskip

A final technique to make everything come together is the use of dimensionality reductions.
We can't quite use the standard bit-sampling and Johnson–Lindenstrauss lemmas, since those may (though unlikely) increase the distance between originally near points.
Instead we use two dimensionality reduction lemmas based on partitioning.
Similarly to~\cite{pagh2016locality} and others, we fix a random permutation.
Then given a vector $x\in\zo^d$ we permute the coordinates and partition into blocks $x_1, \dots, x_{d/B}$ of size $B$.
For some linear distance function, $\dist(x,y)=\dist(x_1,y_1)+\dots+\dist(x_{d/B},y_{d/B})$, which implies that for some $i$ we must have $\dist(x_i,y_i)\le\dist(x,y)B/d$.
Running the algorithm separately for each set of blocks guarantee that we no pair gets mapped too far away from each other, while the randomness of the permutation lets us apply standard Chernoff bounds on how close the remaining points get.

Partitioning, however, doesn't work well if distances are very small, $cr<<d$.
This is because we need $B = \frac{d}{cr}\epsilon^{-2}\log n$ to get the said Chernoff bounds on distances for points at distance $cr$.
We solve this problem by hashing coordinates into buckets of $\approx cr/\epsilon$ and taking the xor of each bucket.
This has the effect of increasing distances and thereby allowing us to partition into blocks of size $\approx\epsilon^{-3}\log n$.
A similar technique was used for dimensionality reduction in~\cite{kushilevitz2000efficient}, but without deterministic guarantees.
The problem is tackled fully deterministically in~\cite{indyk2000dimensionality} using codes, but with the slightly worse bound of $\epsilon^{-4}\log n$.

\smallskip

For the second problem of Braun-Blanquet similarity we also need a way to reduce the dimension to a manageble size.
Using randomized reductions (for example partitioning), we can reduce to $|x\cap y|\sim \log n$ without introducing too many false positives.
However we could easily have e.g. universe size $d=(\log n)^{100}$ and $|x|=|y|=(\log n)^2$, which is much too high a dimension for our splitter technique to work.
There is probably no hope of actually reducing $d$, since increasing $|x|/d$ and $|y|/d$ makes the problem we are trying to solve easier, and such a reduction would thus break LSH lower bounds.

Instead we introduce tensoring technique based on perfect hash functions, which allows us to create Turán Systems with very large universe sizes for very little overhead.

\smallskip

In the process of showing our results, we show a useful bound on the ratio between two binomial coefficients, which may be of separate interest.

\subsection{Notation}

We use $[d] = \{1,\dots,d\}$ as convenient notation sets of a given size.
Somewhat overloading notation, for a predicate $P$, we also use the Iversonian notation $[P]$ for a value that is 1 if $P$ is true and 0 otherwise.

For a set $x\subseteq[d]$, we will sometimes think of it as a subset of the universe $[d]$, and at other times as a vector $x\in\zo^d$, where $x_i=1$ indicates that $i\in x$.
This correspondence goes further, and we may refer to the set size $|x|$ or the vector norm $\|x\|$, which is always the Hamming norm, $\|x\|=\sum_{i=1}^d x_i$.
Similarly for two sets or points $x,y\in\zo^d$, we may refer to the inner product $\langle x,y\rangle=\sum_{i=1}^d x_i y_i$ or to the size of their intersection $|x\cap y|$.

We use $S \times T = \{(s,t) : s\in S, t\in T\}$ for the cross product,
and $x \oplus y$ for symmetric difference (or `xor').
$\mathcal P(X)$ is the power set of $X$, such that $x\subseteq X \equiv x\in\mathcal P(X)$.
$\binom{X}{k}$ denotes all subsets of $X$ of size $k$.

For at set $S\subseteq[d]$ and a vector $x\in\zo^d$, we let $x_S$ be the projection of $x$ onto $S$.
This is an $|S|$-dimensional vector, consisting of the coordinates $x_S=\langle x_i : i\in S\rangle$ in the natural order of $i$.
For a function $f : [a] \to [b]$ we let $f^{-1} : \mathcal P([b]) \to \mathcal P([a])$ be the `pullback' of $f$, such that $f^{-1}(S) = \{i \in [a] \mid f(i)\in S\}$.
For example, for $x\in\zo^a$, we may write $x_{f^{-1}(1)}$ to be the vector $x$ projected onto the coordinates of $f^{-1}(\{1\})$.

Sometimes when a variable is $\omega(1)$ we may assume it is integral, when this is achievable easily by rounding that only perturbs the result by an insignificant $o(1)$ amount.

The functional $\poly(a,b,\dots)$ means any polynomial combination of the arguments, essentially the same set as $(a\cdot b \dots)^{\pm O(1)}$.

\subsection{Organization}

We start by laying out the general framework shared between our algorithms.
We use a relatively common approach to modern near neighbour data structures, but the overview also helps establish some notation used in the later sections.

The second part of section~\ref{sec:overview} describes the main ideas and intuition on how we achieve our results.
In particular it defines the concept of `splitters' and how they may be used to create list-decodable codes for various measures.
The section finally touches upon the issues we encounter on dimensionality reduction, which we can use to an extent, but which is restricted by our requirement of `1-sided' errors.

In sections~\ref{sec:hamming} and~\ref{sec:similarity} we prove the main theorems from the introduction.
The sections follow a similar pattern:
First we introduce a filter family and prove its existence,
then we show a dimensionality reduction lemma and analyze the resulting algorithm.

\section{Overview}\label{sec:overview}

Both algorithms in this paper follow the structure of the Locality Sensitive Filter framework, which is as follows:
For a given universe $U$,
we define a family $\mathcal F$ of `filters' equipped with
a (possibly random) function $F : U \to \mathcal P(\mathcal F)$,
which assigns every point a set of filters.

Typically, $\mathcal F$ will be a generous covering of $U$, and $F(x)$ will be the sets that cover the point $x$.
Critically, any pair $x,y$ that is close/similar enough in $U$ must share a filter, such that $F(X)\cap F(Y)\neq\emptyset$.
Further we will want that pairs $x,y$ that are sufficiently far/dissimilar only rarely share a filter, such that $E[|F(x)\cap F(Y)|]$ is tiny.
%To get good expected performance, some randomness, such as a rotation that does not change distances, is performed to make this only happen with low probability.

To construct the data structure, we are given a set of data points $P\subseteq U$.
We compute $F(x)$ for every $x\in P$
and store the points in a (hash) map $T : \mathcal F \to \mathcal P(P)$.
For any point $x\in P$ and filter $f\in F(x)$, we store $x \in T[f]$.
Note that the same $x$ may be stored in multiple different buckets.

To query the data structure with a point $x\in U$, we compute the distance/similarity between $x$ and every point $y \in \bigcup_{f\in F(x)} T[f]$, returning the first suitable candidate, if any.

There are many possible variations of the scheme, such as sampling $\mathcal F$ from a distribution of filter families.
In case we want a data structure with space/time trade-offs, we can use different $\mathcal F$ functions for data points and query points.
However in this article we will not include these extensions.

We note that while it is easy to delete and insert new points in the data structure after creation, we are going to choose $\mathcal F$ parametrized on the total number of points, $|P|$.
This makes our data structure essentially static, but luckily~\cite{overmars1981worst} have found general, deterministic reductions from dynamic to static data structures.

\subsection{Intuition}

The main challenge in this paper will be the construction of filter families $\mathcal F$ which are: (i) not too large; (ii) have a $F(\cdot)$ function that is efficient to evaluate; and most importantly, (iii) guarantee that all sufficiently close/similar points always share a filter.
The last requirement is what makes our algorithm different from previous results, which only had this guarantee probabilistically.

For concreteness, let us consider the Hamming space problem.
Observe that for very low dimensional spaces, $d=(1+o(1))\log n$, we can afford to spend exponential time designing a filter family.
In particular we can formulate a set cover problem, in which we wish to cover each pair of points at distance $\le r$ with Hamming balls of radius $s$.
This gives a family that is not much larger than what can be achieved probabilistically, and which is guaranteed to work.
Furthermore, this family has sublinear size ($n^{o(1)}$), making $F(x)$ efficient to evaluate, since we can simply enumerate all of the Hamming balls and check if $x$ is contained.

The challenge is to scale this approach up to general $d$.

Using a standard approach of randomly partitioning the coordinates, we can reduce the dimension to $(\log n)^{1+\epsilon}$.
This is basically dimensionality reduction by bit sampling, but it produces $d/\log n$ different subspaces, such that for any pair $x,y$ there is at least one subspace in which their distance is not increased.
We are left with a gap from $(\log n)^{1+\epsilon}$ down to $\log n$.
Bridging this gap turns out to require a lot more work.
Intuitively we cannot hope to simply use a stronger dimensionality reduction, since $\log n$ dimensions only just fit $n$ points in Hamming space and would surely make too many non-similar points collide to be effective.

A natural idea is to construct higher-dimensional filter families by combining multiple smaller families.
This is a common technique from the first list decodable error correcting codes, for example~\cite{elias1957list}:
Given a code $\mathcal C\subseteq\zo^d$ with covering radius $r$, we can create a new code $\mathcal C_2\subseteq\zo^{2d}$ of size $|\mathcal C|^2$ with covering radius $2r$ by taking every pair of code words and concatenating them.
Then for a given point $x\in\zo^{2d}$ we can decode the first and last $d$ coordinates of $x=x_1x_2$ separately in $\mathcal C$.
This returns two code words $c_1, c_2$ such that $\dist(x_1,c_1)\le r$ and $\dist(x_2,c_2)\le r$.
By construction $c_1c_2$ is in $\mathcal C_2$ and $\dist(x_1x_2, c_1c_2)\le 2r$.

This combination idea gives is nice when it applies.
When used with high quality inner codes, the combined code is close to optimal as well.
In most cases the properties of $\mathcal C$ that we are interested in won't decompose as nicely.
With the example of our Hamming ball filter family, consider $x,y\in\zo^{2d}$ with distance $\dist(x,y)=r$.
If we split $x=x_1x_2$ an $y=y_1y_2$ we could decode the smaller vectors individually in a smaller family,
however we don't have any guarantee on $\dist(x_1,y_1)$ and $\dist(x_2,y_2)$ individually, so the inner code might fail to return anything at all.

To solve this problem, we use a classic tool for creating combinatorial objects, such as our filter families, called `splitters'.
Originally introduced by~\cite{mairson1983program,naor1995splitters} they are defined as follows:
\begin{definition}[Splitter]\label{defn:splitter}
   A $(B, l)$-splitter $H$ is a family of functions from
   $\{1,\dots,B\}$ to $\{1,\dots,l\}$ such that for all $S \subseteq \{1,\dots,B\}$, there is a $h \in H$ that splits $S$ perfectly,
   i.e., into equal-sized parts $h^{-1}(j) \cap S$, $j = 1, 2, \dots, l$.
   (or as equal as possible, if $l$ does not divide $|S|$).
\end{definition}
%TODO: In longer version, include this argument
The size of $H$ is at most $B^l$, and using either a generalization by~\cite{alon2006algorithmic} or a simple combinatorial argument, it is possible to ensure that the size of each part $|h^{-1}(j)|$ equals $B/l$ (or as close as possible).

We now explain how splitters help us combine filter families.
Let $H$ be a splitter from $\{1,\dots,2d\}$ to $\{1,2\}$.
For any $x,y\in\zo^{2d}$ we can let $S$ be the set of coordinates on which $x$ and $y$ differ.
Then there is a function $h\in H$ such that $|h^{-1}(1)\cap S|=|h^{-1}(2)\cap S|=|S|/2$.
(Or as close as possible if $|S|$ is odd.)
If we repeat the failed product combination from above for every $h\in H$ we get a way to scale our family from $d$ to $2d$ dimensions, taking the size from $|\mathcal F|$ to $(2d)^2|\mathcal F|^2$.
That is, we only suffer a small polynomial loss.
In the end it turns out that the loss suffered from creating filter families using this divide and conquer approach can be contained, thus solving our problem.

An issue that comes up, is that the `property' we are splitting (such as distance) can often be a lot smaller than the dimensionality $d$ of the points.
In particular this original dimensionality reduction may suffer an overhead factor $d/|S|$, which could make it nearly useless if $|S|$ is close to $1$.
To solve this problem, both of our algorithms employ special half-deterministic dimensionality reductions, which ensures that the interesting properties get `boosted' and end up taking a reasonable amount of `space'.
These reductions are themselves not complicated, but they are somewhat non-standard, since they can only have a one sided error.
For example for Hamming distance we need that the mapped distance is never larger than its expected value, since otherwise we could get false negatives.

For Hamming distance our dimension reduction works by hashing the coordinates randomly from [d] to [m] taking the xor of the coordinates in each bucket.
This is related to the $\beta$-test in~\cite{kushilevitz2000efficient}.
The idea is that if $x$ and $y$ are different in only a few coordinates, then taking a small random group of coordinates, it is likely to contain at most one where they differ.
If no coordinates differ, then after taking the xor the result will still be the same, but if exactly one (or an odd number) of coordinates differ, the resulting coordinate will be different.

\medskip

For set similarity things are a bit more hairy.
There is no data independent dimensionality reduction that can reduce the size of the domain.
In fact this would break the lower bounds of e.g.~\cite{tobias2016}.
Instead we come up with a new construction based on perfect hash functions, which greatly increases the number of filters needed, but only as much as we can afford given the easier sub-problems.

The idea can be motivated as follows:
Suppose you have to make a family of sets $\mathcal T \subseteq \mathcal P([n])$ of size $r$, such that for each each set $K \subseteq [n]$ of size $|K|=k$ there is an $R\in\mathcal T$ such that $R\subseteq K$.
Then you might try to extend this to the domain $[2n]$ as follows:
For each $R\in\mathcal T$ and each $b\in\zo^r$, make a new set $R' = \{i+n b_i : i\in R\}$ (where $b_i$ is padded appropriately).
This creates $2^r|\mathcal T|$ new sets, which can be shown to have the property, that for any set $K\subseteq[2n]$ of size $|K|=k$, there is an $R'$ such that $R'\subseteq K$.
That is as long as $K\cap(K-n)=\emptyset$, since then we can consider $R\in \mathcal T$ such that $R\subseteq (K\bmod n)$.
That is $R$ is a subset of $K$ `folded` into the first $[n]$ elements, and one of the $R'$ will be a subset of $K$.

Because of the requirement that $|K\bmod n|=k$ we need to use perfect hashing as a part of the construction.
However for non-Las Vegas algorithms, a similar approach may be useful, simply using random hashing.

\section{Hamming Space Data Structure}\label{sec:hamming}

We will give an efficient filter family for LSF in Hamming space.
Afterwards we will analyze it and choose the most optimal parameters, including dimensionality reduction.

\begin{lemma}\label{lem:hamcode}
   For every choice of parameters $B, b\in\mathbb N$, $b\le B$, $0<r<B/2$ and $s^2=O(B/\sqrt{b})$, there exists a code $\mathcal C\subseteq \zo^B$ of size $|\mathcal C| = \poly(B^{B/b})\, \exp(\frac{s^2}{2(1-r/d)})$
   with the following properties:
   \begin{enumerate}
      \item Given $x\in\zo^B$ we can find a subset
         $C(x) \subseteq \{c\in \mathcal C : \dist(x,c)\le B/2-s\sqrt{B}/2\}$
         in time $|C(x)| + \poly(B^{B/b},e^{s^2b/B})$.
      \item For all pairs $x,y\in\zo^B$ with $\dist(x,y)\le r$ there is some common nearby code word $c\in C(x)\cap C(y)$.
      \item The code requires $4^b\poly(B^{B/b}, e^{s^2b/B})$ time for preprocessing and $\poly(B^{B/b}, e^{s^2b/B})$ space.
   \end{enumerate}
\end{lemma}
%
%The requirement $s=O(B^{1/2}/b^{1/4})$ is purely technical to allow a simple normal approximation to the size of hamming balls.
%Our reduction 
%to be tight, but there is in fact an interesting regime of `low dimensionality' algorithms that take $s$ larger. See~\cite{becker2016new} for details.

Note that we don't actually guarantee that our `list-decoding' function $C(x)$ returns \emph{all} nearby code words, just that it returns enough for property (2) which is what we need for the data structure.
This is however not intrinsic to the methods and using a decoding algorithm similar to~\cite{becker2016new} would make it a `true' list-decoding.

\begin{proof}
   We first show how to construct a family for $\zo^b$, then how to enlarge it for $\zo^B$.
   We then show that it has the covering property and finally the decoding properties.
   In order for our probabilistic arguments to go through, we need the following lemma, which follows from Stirling's Approximation:

   \begin{restatable}{lemma}{intt}
   \label{lem:int}
         For $t=\frac d2-\frac{s\sqrt d}2$, $1\le s \le d^{1/4}/2$ and $r<d/2$,
         Let $x,y\in\zo^d$ be two points at distance $\dist(x,y)=r$, and
         let $I = |\{z\in\zo^d : \dist(z,x) \le t, \dist(z,y) \le t\}|$
         be the size of the intersection of two hamming balls around $x$ and $y$ of radius $t$,
         then
         \begin{align*}
            \frac{7}{8d} \exp\left(\tfrac{-s^2}{2(1-r/d)}\right)
            \le
            I\, 2^{-d} \le \exp\left(\tfrac{-s^2}{2(1-r/d)}\right)
         \end{align*}
   \end{restatable}
   Proof is in the appendix.

   Let $s' = s\sqrt{b/B}$.
   Consider any two points $x,y\in\zo^b$ with distance $\le (r/d) b$.
   If we choose $a \in \zo^b$ uniformly at random, by lemma~\ref{lem:int} we have probability $p=\poly(b)\exp(\frac{-s'^2}{2(1-r/d)})$ that both $x$ and $y$ have distance at most $t=b/2 - s'\sqrt{b/4}$ with $c$.
   By the union bound over pairs in $\zo^b$, if we sample $p^{-1}b\log2$ independent $a$s, we get constant probability that some $a$ works for every pair.
   We can verify that a set $A$ of such filters indeed works for every pair in time $4^b |A|$.
   By repeatedly sampling sets $A$ and verifying them, we get a working $A$ in expected $O(1)$ tries.\footnote{
   The randomness is not essential, and we could as well formulate a set cover instance and solve it using the greedy algorithm, which matches the probabilistic construction up to a log factor in size and time.}

   Next we define $\mathcal C$.
   We build a splitter, that is a set $\Pi$ of functions $\pi : [B] \to [B/b]$, such that for every set $I\subseteq[B]$ there is a $\pi\in\Pi$ such that $\lfloor|I|b/B\rfloor \le |\pi^{-1}(j)\cap I|\le \lceil |I|b/B\rceil$ for $j=1, \dots, B/b$.
   By the discussion after definition~\ref{defn:splitter}, such a set of size $\poly(B^{B/b})$ exists and can be constructed in time and space proportional to its size.
   Implicitly we can then define $\mathcal C$ by making one code word $c\in\zo^B$ for every $\pi\in\Pi$ and $1\le j_1, \dots, j_{B/b}\le |A|$, satisfying the requirement that $c_{\pi^{-1}(j_k)} = A_{j_k}$ for $k=1\dots B/b$.
   That is, for a given set of rows of $A$ and a split of $[B]$, we combine the rows into one row $c$ such that each row occupies a separate part of the split.
   Note that this is possible, since splitter has all partitions of equal size, $b$.
   The created family then has size $|\mathcal C| =|\Pi||A|^{B/b}
   = \poly(B^{B/b})\exp(\frac{-s^2}{2(1-r/d)})$ as promised.
   Since the only explicit work we had to do was finding $A$, we have property (3) of the lemma.

   We define the decoding function $C(x)\in\mathcal C$ for $x\in\zo^B$ with the following algorithm:
   For each $\pi\in\Pi$ compute the inner decodings $A_j = \{a\in A : \dist(x_{\pi^{-1}(j)}, a) \le b/2-s\sqrt{b}/2\}$ for $j=1,\dots,B/b$.
   Return the set of all concatenations in the product of the $A_j$'s: $C(x) = \{a_1 \| a_2 \| \dots \| a_{B/b} : a_1 \in A_1, \dots\}$.
   Computing the $A_j$'s take time $(B/b)|A|$, while computing and concatenating the product takes linear time in the size of the output.
   This shows property (1).

   Finally for property (2), consider a pair $x,y\in\zo^B$ of distance $\le r$.
   Let $I$ be the set of coordinates on which $x$ and $y$ differ.
   Then there is a function $\pi\in\Pi$ such that $x$ and $y$ differ in at most $|I|b/B = rb/B$ coordinates in each subset $\pi^{-1}(1),\dots,\pi^{-1}(B/b)\subseteq [B]$.
   Now for each pair of projected vectors $x_{\pi^{-1}(1)}, y_{\pi^{-1}(1)}, \dots$ (let's call them $x_1, y_1, \dots$) there is an $a_j\in A$ such that $\dist(a_j,x_j) \le b/2-s'\sqrt{b}/2$ and $\dist(a_j,y_j) \le b/2-s'\sqrt{b}/2$.
   This means that $x$ and $y$ must both have distance at most $(b/2-s'/2)B/b=B/2-s\sqrt{B}/2$ to that $c\in\mathcal C$ which has
   $c_{\pi^{-1}(j)} = a_{j}$ for $j=1\dots B/b$.
   By the same reasoning, this $c$ will be present in both $C(x)$ and $C(y)$, which proves the lemma.
\end{proof}

Returning to the problem of near neighbour search in $\zo^d$, it is clear from the $4^b\poly(B^{B/b})$ construction time of the above family, that it will not be efficient for dimension $B = (\log n)^{\omega(1)}$.
For this reason we will apply the following dimensionality reduction lemma:
\begin{lemma}\label{lem:embed}
   Given $d\ge cr>r\ge 1$ and $\epsilon,\delta>0$, define $B=27\epsilon^{-3}\log 1/\delta$ and $m= 3cr/\epsilon$ and assume $\delta^{-1}\ge m$,
   then there is a random set $F$ of at most $S=m/B$ functions $f : \zo^d \to \zo^{B}$ with the following properties for every $x,y\in\zo^d$:
   \begin{enumerate}
      \item With probability 1, there is at least one $f\in F$ st.:
         \begin{align*}
            \dist(f(x),f(y))\le \dist(x,y)/S.
         \end{align*}
      \item If $\dist(x,y)\ge cr$ then for every $f\in F$ with probability at least $1-\delta$:
         \begin{align*}
            \dist(f(x),f(y))\ge (1-\epsilon)cr/S.
         \end{align*}
   \end{enumerate}
\end{lemma}
The idea is to randomly throw the $d$ coordinates in $m=3cr/\epsilon$ buckets, (xor-ing the value if more than two coordinates land in the same group.)
For two points with $\le cr$ differences, this has the effect of rarely colliding two coordinates at which the points disagree, thus preserving distances.
It further has the effect of changing the relative distances from arbitarily low $r/d$ to something around $\epsilon$, which allows partitioning the coordinates into groups of around $\epsilon^{-3}\log1/\delta$ coordinates using Chernoff bounds.
\begin{proof}
   To prove lemma~\ref{lem:embed} first notice that
   if $B\ge d$ we can use the identity function and we are done.
   If $B\ge m$, then we can duplicate the vector $\lceil m/B\rceil = O(\epsilon^{-2}\log 1/\delta)$ times.
   Also, by adjusting $\epsilon$ by a constant factor, we can assume that $B$ divides $m$.

   For the construction, pick a random function $h : [d] \to [m]$.
   Define $g : \zo^d \to \zo^m$ by `xor'ing the contents of each bucket, $g(x)_i = \bigoplus_{j\in h^{-1}(i)}x_j$,
   and let $f_{i}(x) = g(x)_{(iB,(i+1)B]}$ for $i=0\dots m/B$ be the set of functions in the lemma.
   We first show that this set has property (1) and then property (2).

   Observe that $g$ never increases distances, since for any $x,y\in\zo^d$ the distance
   \begin{align*}
      \dist(g(x),g(y))
      = \sum_{i=1}^m \left[\bigoplus_{j\in h^{-1}(i)} x_j \neq \bigoplus_{j\in h^{-1}(i)} y_j\right]
   \end{align*}
   is just $\sum_{i=1}^m\left(\sum_{j\in h^{-1}(i)} [x_j \neq y_j] \bmod 2\right)$ which is upper bounded by the number of coordinates at which $x$ and $y$ differ.
   By averaging, there must be one $f_i$ such that $\dist(f_i(x),f_i(x))\le \dist(g(x),g(y)) B/m \le \dist(x,y) /S$.

   For the second property, let $R=\dist(x,y)\ge cr$ and let $X_1, \dots, X_m$ be the random number of differences between $x$ and $y$ in each bucket under $h$.
   Let $Y_1, \dots, Y_m$ be iid. Poisson distributed variables with mean $\lambda=\E X_1 = R/m\ge\epsilon/3$.
   We use the the Poisson trick from~\cite{mitzenmacher2005probability} theorem 5.7:
   \begin{align*}
      \Pr[\sum_{i=1}^B (X_i \bmod 2) < x]
      \le e\sqrt{m} \Pr[\sum_{i=1}^B (Y_i \bmod 2) < x]
      .
   \end{align*}
   The probability $\Pr[Y \bmod 2\equiv1]$ that a Poisson random variable is odd is $(G(1)-G(-1))/2$ where $G(z)=\sum_i \Pr[Y=i]z^i = e^{\lambda(z-1)}$.
   This gives us the bound $\Pr[Y_i \bmod 2\equiv1]=(1-e^{-2\lambda})/2\ge(1-e^{-2\epsilon/3})/2\ge (1-\epsilon/3)\epsilon/3$.
   We can then bound the probability of an $f_i$ decreasing distances too much, using a Chernoff bound ($\Pr[Z\le x]\le\exp(-D[x/B\mid p]B)$):
   \begin{align*}
      &\Pr[\dist(f_i(x),f_i(y)) \le (1-\epsilon) cr/S]\\
      &\quad\le e\sqrt{m} \exp(-D[(1-\epsilon)\epsilon/3\mid (1-\epsilon/3)\epsilon/3]B)\\
      &\quad\le e\sqrt{m} \exp(-2\epsilon^3 B/27).
   \end{align*}
   Since $cr/S=cr B/m=B\epsilon/3$.
   Here $D[\alpha\mid\beta] = \alpha \log\frac{\alpha}{\beta} + (1-\alpha) \log\frac{1-\alpha}{1-\beta} \ge (\alpha-\beta)^2/(2\beta)$ is the Kullback–Leibler divergence.
   For our choice of $B$ the error probability is then $e\sqrt{m}\delta^2$ which is less than $\delta$ by our assumptions.
   This proves the lemma.
\end{proof}

Using lemma~\ref{lem:embed} we can make at most $3cr/\epsilon= O(d/\epsilon)$ data structures, as described below, and be guaranteed that in one of them, we will find a near neighbour at distance $r' = r/S=\epsilon/(3c)B$.
In each data structure we will have to reduce the distance $cr'$, at which we expect far points to appear, to $cr'(1-\epsilon)$.
This ensures we see at most a constant number of false positives in each data structure, which we can easily filter out.
For $\epsilon=o(1)$ this change be swallowed by the approximation factor $c$, and won't significantly impair our performance.

When using the filter family of lemma~\ref{lem:hamcode} for LSF, the time usage for queries and inserting points is dominated by two parts:
1) The complexity of evaluating $C(x)$, and 2) The expected number of points at distance larger than $cr'(1-\epsilon)$ that falls in the same filter as $x$.

By randomly permuting and flipping the coordinates, we can assume that $x$ is a random point in $\zo^d$.
The expected time to decode $C(x)$ is then
\begin{align*}
   &\E|C(x)| + \poly(B^{B/b}, e^{s^2b/B})\\
   &= |\mathcal C|\Pr_x[0\in C(x)] + \poly(B^{B/b}, e^{s^2b/B})\\
   &\le \poly(B^{B/b}, e^{s^2b/B})\exp\left(\tfrac{s^2}{2(1-r'/B)}-\tfrac{s^2}{2}\right)
   .
\end{align*}

For estimating collisions with far points, we can similarly assume that $x$ and $y$ are random points in $\zo^d$ with fixed distance $cr'(1-\epsilon)$:
\begin{align*}
   &\E|\{y \in P : C(x)\cap C(y)\neq \emptyset\}|\\
   &\le n\, |\mathcal C|\, \Pr_{x,y}[0\in C(x), 0\in C(y)]\\
   &\le B^{O(B/b)} \exp\left(\tfrac{s^2}{2(1-r'/B)}-\tfrac{s^2}{2(1-c(1-\epsilon)r'/B)}+\log n\right)\\
   &= B^{O(B/b)} \exp\left(\tfrac{s^2}2(\tfrac1{1-r'/B}-\tfrac{1}{1-cr'/B}+O(\epsilon))+\log n\right)
   .
\end{align*}
Finally we should recall that constructing the data structures takes time $4^b \poly(e^{s^2b/B})$ plus $n$ inserts.

We now choose the parameters:
\begin{align*}
   s^2/2 &= \tfrac{1-cr'/B}{cr'/B}\log n,
         &B &= 27\epsilon^{-3}\log n,\\\
   b &= \log_4 n,
     &\epsilon &= (\log n)^{-1/4}.
\end{align*}
This makes the code construction time $n^{1+o(1)}$ while evaluating $C(x)$ and looking at far points takes expected time at most $n^{1/c+\tilde O(\log n)^{-1/4}}$.
To use lemma~\ref{lem:hamcode} we have to check that $s^2=O(B/\sqrt{b})=O((\log n)^{5/4})$, but $s^2/2 = \frac{1-\epsilon/3}{\epsilon/3}\log n = (\log n)^{5/4}(1-o(1))$ so everything works out.
This shows theorem~\ref{thm:ham}.

To get the result of corollary~\ref{cor:ham}, we just need to substitute the dimensionality reduction lemma~\ref{lem:embed} for a simple partitioning approach. (Lemma \ref{lem:partitioning} below.)
The idea is that of~\cite{pagh2016locality} which is to randomly partition the $d$ coordinates in $B$ parts and run the algorithm on those.
The important point is that in this case $r'/B = r/d$, so the relative distance is not decreased.
We choose parameters
\begin{align*}
   s^2/2 &= \tfrac{1-cr/d}{cr/d}\log n
         &B &= O(\epsilon^{-2}(cr/d)^{-1}\log n),\\
   b &= \log_4 n,
     &\epsilon &= (\log n)^{-1/3}.
\end{align*}
This again makes this makes the code construction time $n^{1+o(1)}$ while evaluating $C(x)$ and looking at far points takes time $n^{\tfrac{1-c\delta}{c(1-\delta)}+\tilde O(\log n)^{-1/3}d/r}$ as in the corollary.
Again we need need to check that $s^2=O(B/\sqrt b)=O((\log n)^{7/6})$.
This works as long as $r/d=\Omega((\log n)^{-1/6})$, which is the assumption of the corollary.

\begin{lemma}\label{lem:partitioning}
   For any $d \ge r \ge 1$ and $\epsilon>0$ there is a set $F$ of $d/B$ functions, $f : \zo^d \to \zo^B$, where $B=2\epsilon^{-2}d/(cr)\log n$, such that:
   \begin{enumerate}
      \item With probability 1, there is at least one $f\in F$ st.:
         \begin{align*}
            \dist(f(x),f(y))\le \dist(x,y)\,B/d.
         \end{align*}
      \item If $\dist(x,y)\ge cr$ then for every $f\in F$ with probability at least $1-1/n$:
         \begin{align*}
            \dist(f(x),f(y))\ge (1-\epsilon)cr\,B/d.
         \end{align*}
   \end{enumerate}
\end{lemma}
\begin{proof}
   Fix a random permutation.
   Given $x\in\zo^d$, shuffle the coordinates using the permutation.
   Let $f_1(x)$ be the first $B$ coordinates of $x$, $f_2(x)$ the next $B$ coordinates and so forth.
   For any $y\in\zo^d$, after shuffling, the expected number of differences in some block of $B$ coordinates is $\dist(x,y)B/d$.
   Then the first property follows because $\sum_i\dist(f_i(x),f_i(y))=\dist(x,y)$ so not all distances can be below the expectation.
   The second property follows from the Chernoff/Hoeffding bound~\cite{hoeffding1963probability}.
\end{proof}

\section{Set Similarity Data Structure}\label{sec:similarity}

We explore the generality of our methods, by making a Las Vegas version of another very common LSH data structure.
Recall the theorem we are trying to prove, from the introduction:

\begin{theorem*}
   Given a set $P$ of $n$ subsets of $[d]$, define
   the Braun-Blanquet similarity
   $\simi(x,y) = |x\cap y|/\max(|x|,|y|)$ on the elements of $P$.
   For every choice of $0 < b_2 < b_1 < 1$ there is a data structure on $P$ that supports the following query operation:

   On input $q \subseteq [d]$, for which there exists a set $x \in P$ with $\simi(x,q) \ge b_1$, return an $x' \in P$ with $\simi(x', q) > b_2$.
   The data structure uses expected time
   $dn^\rho$ per query,
   can be constructed in expected time $d n^{1+\rho}$,
   and takes expected space $n^{1+\rho} + dn$ where
   $\rho = \frac{\log1/b_1}{\log1/b_2} +\, \scriptstyle \hat O(1/\sqrt{\log n})$.
\end{theorem*}

By known reductions~\cite{tobias2016} we can focus on the case where all sets have the same weight, $w$, which is known to the algorithm.
This works by grouping sets in data structures with sets of similar weight and uses no randomization.
% This may not be entirely true, given that they also reduce the dimension size
% something that shouldn't be doable.
The price is only a $(\log n)^{O(1)}$ factor in time and space, which is dominated by the $n^{\hat O(1/\sqrt{\log n})}$ factor in the theorem.

When two sets have equal weight $w$, being $b$-close in Braun-Blanquet similarity coresponds exactly to having an intersection of size $b w$.
Hence for the data structure, when trying to solve the $(b_1,b_2)$-approximate similarity search problem, we may assume that the intersections between the query and the `close' sets we are interested in are at least $b_1 w$, while the intersections between the query and the `far' sets we are not interested in are at most $b_2 w$.

The structure of the algorithm follows the LSF framework as in the previous section.
A good filter family for set similarity turns out to be the $r$-element blocks of a Turán system.
This choice is inspired by~\cite{tobias2016} who sampled subsets with replacement.
\begin{definition}[\cite{turan1961research,colbourn2006handbook}]
A Turán $(n, k, r)$-system is a collection of $r$-element subsets, called `blocks', of an $n$ element set $X$ such that every $k$ element subset of $X$ contains at least one of the blocks.
\end{definition}

Turán systems are well studied on their own, however all known general constructions are either only existential or of suboptimal size.
The following lemma provides the first construction to tick both boxes, and with the added benefit of being efficiently decodable.

An interesting difference between this section and the last, is that we don't know how to do a dimensionality reduction like we did for hamming distance.
Instead we are (luckily) able to make an efficiently decodable filter family even for very large dimensional data points.

\begin{lemma}\label{lem:turan}
   For every $n, k, r$, where $n > k > r^{3/2}$, there exists a Turán $(n,k,r)$-system, $\mathcal T$, of
   % expected size $E|\mathcal T|
   % It doesn't have to be expected size.
   size $|\mathcal T|\le (n/k)^r \, e^\chi$
   where $\chi=O(\sqrt{r}\log r+\log k+\log\log n)$
   with the following properties:
   \begin{enumerate}
      \item Correctness: For every set $K\subseteq[n]$ of size at least $k$,
         there is at least one block $R\in\mathcal T$ such that $R\subseteq K$.
      \item Efficient decoding: Given a set $S\subseteq [n]$, we can find all the blocks it contains $T(S) = \{R\in \mathcal T : R \subseteq S\}$ in time $|S||T(S)| + e^\chi$.
         Furthermore, $T(S)$ has expected size $\le (|S|/k)^r e^\chi$.
      \item Efficient construction: The system is constructible, implicitly, in $e^{r(1+o(1))}$ time and space.
   \end{enumerate}
\end{lemma}

Notes:
A simple volume lower bound shows that an $(n,k,r)$-system must have at least ${n\choose r}/{k\choose r}\ge (n/k)^r$ blocks, making our construction optimal up the factor $e^\chi$.
Using the sharper bound
${n\choose r}/{k\choose r}\approx (n/k)^r \exp(\frac{r^2}{2k})$
from lemma~\ref{lem:ratio}, we get that the factor
$\exp(\Omega(\sqrt{r}))$ is indeed tight for $k = O(r^{3/2})$.

The size of the system is in `expectation', which is sufficient for our purposes, but is in fact fairly easy to fix.
On the other hand, the `expectation' in the size of sets $T(S)$ seems harder to get rid of, which is the reason why the data strcuture is Las Vegas and not deterministic.

\subsection{The algorithm}

We continue to describe the algorithm and proving theorem~\ref{thm:sim} using the lemma.
The proof of the lemma is at the end and will be the main difficulty of the section.

As discussed, by the reduction of~\cite{tobias2016} we can assume that all sets have weight $w$, intersections with close sets have size $\ge b_1 w$ and intersections with far sets have size $\le b_2 w$.
The data structure follows the LSF scheme as in the previous section.
For filters we use a Turán $(d, b_1w, \frac{\log n}{\log1/b_2})$ design, constructed by lemma~\ref{lem:turan}.
Note that if $b_1w < (\frac{\log n}{\log1/b_2})^{3/2}$ ($k<r^{3/2}$ in the terms of the lemma), we can simply concatenate the vectors with themselves $O(\log n)$ times.
If $b_1w \le \frac{\log n}{\log1/b_2}$ we can simply use all the $\binom{d}{b_1w}$ sets of size $b_1w$ as a Turán $(d,b_1w,b_1w)$ system and get a fast deterministic data structure.

As in the previous section, given a dataset $P$ of $n$ subsets of $[d]$, we build the data structure by decoding each set in our filter system $\mathcal T$.
We store pointers from each set $R\in\mathcal T$ to the elements of $P$ in which they are a contained.
By the lemma, this takes time
$n(w(w/k)^r+1)e^\chi = wn(1/b_1)^{\frac{\log n}{\log1/b_2}}e^\chi \le dn^\rho$, while expected space usage is $n(w/k)^re^\chi + e^{r(1+o(1))} + dn = n^\rho + dn$ as in the theorem.
Building $\mathcal T$ takes time $e^{r(1+o(1))} = n^{(1+o(1))/\log1/b_2} = n^{1+o(1)}$.

Queries are likewise done by decoding the query set $x$ in $\mathcal T$ and inspecting each point $y\in P$ for which there exists $R\in\mathcal T$ with $R\subseteq y$, until a point $y'$ with $\simi(x,y') > b_2$ is found.
Let's call this the candidate set of $x$.
The expected number of false positive points in the candidates is thus
\begin{align*}
   \sum_{y\in P}\E[|\{R \in \mathcal T : R\subseteq x\cap y\}|]
   =\sum_{y\in P}\E[|T(x\cap y)|]
   \le n (b_2w/(b_1w))^{\frac{\log n}{\log1/b_2}} e^\chi
   = n^{\rho}.
\end{align*}
Computing the actual similarity takes time $w$, so this step takes time at most $wn^{\rho}\le dn^{\rho}$ .
We also have to pay for actually decoding $T(x)$, but that takes time
$w (w/(b_1 w))^{\frac{\log n}{\log1/b_2}}e^\chi + e^\chi \le d n^{\rho}$ as well.

Finally, to see that the algorithm is correct, if $\simi(x,y) \ge b_1$ we have $|x\cap y|\ge b_1 w$, and so by the Turán property of $\mathcal T$ there is $R\in\mathcal T$ such that $R\subseteq x\cap y$ which implies $R\subseteq x$ and $R\subseteq y$.
This shows that there will always be at least one true high-similarity set in the candidates, which proves theorem~\ref{thm:sim}.

\subsection{The proof of lemma~\ref{lem:turan}}

\begin{proof}
   We first prove the lemma in four parts, starting with a small design and making it larger step by step.
   To more easily describe the small designs,
   define $a = kr^{-3/2}\log(r^{3/2})$
   and $b = \sqrt{r}$.
   The steps are then
   \begin{enumerate}
      \item Making a $(k^2/(a^2b), k/(ab), r/b)$ using brute force methods.
      \item Use splitters to scale it to $((k/a)^2, k/a, r)$.
      \item Use perfect hashing to make it an $(n/a, k/a, r)$ design.
      \item Use partitioning to finally make an $(n, k, r)$ design.
   \end{enumerate}
   We first prove the lemma without worrying about the above values being intergers.
   Then we'll show that each value is close enough to some integer that we can hide any excess due to approximation in the loss term.

   The four steps can also be seen as proofs of the four inequalities:
   \begin{align*}
      T(n, k, r) & \le \tbinom{n}{r}/\tbinom{k}{r}\, (1 + \log\tbinom{n}{k}),\\
      T(c n, c k, c r) & \le \tbinom{cn}{c}\, T(n,k,r)^c,\\
      T(c k^2, k, r) & \le (k^4 \log ck^2) c^r\, T(k^2, k, r),\\
      T(c n, c k, r) & \le c\, T(n, k, r).
   \end{align*}
   where the $c$ are arbitrary integer constants $>0$.

   \paragraph{1.}
   For convenience, define
   $n' = k^2/(a^2b)$, $k' = k/(ab)$, $r'=r/b$ and assume they are all intergers.
   Probabilistically we can build a Turán $(n',k',r')$-system, $\mathcal T^{(n')}$, by sampling
   \begin{align*}
      \ell
      = \tbinom{n'}{r'}\big/\tbinom{k'}{r'}(1+\log\tbinom {n'}{k'})
      \le (n'/k')^{r'} e^{r'^2/k'} (1+k'\log(en'/k'))
      %\le (n'/k')^{r'} r\sqrt{\log m} (\log m)(1-o(1))
      = (n'/k')^{r'} r^{5/2}(1+o(1))
   \end{align*}
   independent size $r'$-sets.
   (Here we used the bound on $\tbinom{n'}{r'}/\tbinom{k'}{r'}$ from lemma~\ref{lem:ratio} in the appendix.)
   For any size $k'$ subset, $K$, the probability that it contains none of the $r'$-sets is
      $\left(1-\tbinom {k'}{r'}\big/\tbinom {n'}{r'}\right)^\ell
      \le e^{-1}/\tbinom {n'}{k'}$.
      Hence by the union bound over all $\tbinom{n'}{k'}$ K-sets, there is constant probability that they all contain an $r'$-set, making our $\mathcal T^{(n')}$ a valid system.

   We can verify the correctness of a sampled system, naiively, by trying iterating over all $\tbinom{n'}{k'}$ K-sets, and for each one check if any of the R-sets is contained.
   This takes time bounded by
   \begin{align*}
      \tbinom{n'}{k'}\ell
      &\le
      %(ek/a)^{k/(ab)}(kb/a)^{r/b}
      (en'/k')^{k'}(n'/k')^{r'}r^{5/2}(1+o(1))
      \\&=
      %e^{r\frac{1+\log(r^{3/2})-\log\log(r^{3/2})}{\log(r^{3/2})}}
      \left(\tfrac{e r^{3/2}}{\log(r^{3/2})}\right)^\frac{r}{\log(r^{3/2})}
      \left(\tfrac{r^2}{\log(r^{3/2})}\right)^{\sqrt r}r^{O(1)}
      \\&= e^{r+O(r/\log r)}
   \end{align*}
   as in the preprocessing time promised by the lemma.
   Since we had constant success probability, we can repeat the above steps an expected constant number of times to get a correct system.

   Notice that the system has a simple decoding algorithm of brute-force iterating through all $\ell$ sets in $\mathcal T^{(n')}$.

   \paragraph{2.}

   To scale up the system, we proceed as in the previous section, by taking a splitter, $\Pi$, that is a set of functions $\pi : [bn'] \to [b]$ such that for any set $I\subseteq[bn']$ there is a $\pi\in\Pi$ such that
   \begin{align*}
       \lfloor |I|/b\rfloor \le
      |\pi^{-1}(j)\cap I| \le \lceil |I|/b \rceil \enspace \text{ for }j=1, \dots, b.
   \end{align*}
   In other words, each $\pi$ partitions $[bn']$ in $b$ sets $[bn'] = \pi^{-1}(1)\cup\dots\pi^{-1}(b)$ and for any subset $I\subseteq[bn']$ there is a partition which splits it as close to evenly as possible.
   We discuss the constructions of such sets of functions in the appendix.

   %When building subsets of $[n']$ or $[bn']$ it is useful to think of the splitter functions as binary indicator vectors in $\zo^{n'}$ or $\zo^{bn'}$.

   For each $\pi\in\Pi$, and distinct $i_1, \dots, i_{b}\in[|\mathcal T^{(n')}|]$, we make a $br'$-set, $R\subseteq[bn']$, which we think of as an indicator vector $\in\zo^{bn'}$, such that $R_{\pi^{-1}(j)} = \mathcal T^{(n')}_{i_j}$ for $j=1\dots b$.
   That is, the new block, restricted to $\pi^{-1}(1), \pi^{-1}(2), \dots$, will be equal to the $i_1$th, $i_2$th, $\dots$ blocks of $\mathcal T^{(n')}$.
   Another way to think of this is that we take the $i_1$th, $i_2$th, $\dots$ blocks of $\mathcal T^{(n')}$ considered as binary vectors in $\{0,1\}^{n'}$ and combine them to a $bn'$ block `spreading' them using $\pi$.

   The new blocks taken together forms a family $\mathcal T^{(bn')}$ of size
   \begin{align*}
      |\mathcal T^{(bn')}| =|\Pi||\mathcal T^{(n')}|^{b}
      = \tbinom{bn'}{b}[(n'/k')^{r'}r^{O(1)}]^b
      \le (en')^b (n'/k')^{br'} r^{O(b)}
      = (n'/k')^{br'} r^{O(b)},
   \end{align*}
   where we used only the trivial bound $\binom{n}{k}\le(en/k)^k$ and the fact that $n' = r^{O(1)}$.
   %Note that we now have an overhead of $r^{O(b)} \approx e^{\sqrt{r}}$, meaning that this part of the constructions, which is the least efficient of the four parts, explains nearly all the 

   \vspace{0.5em}

   We now show correctness of $\mathcal T^{(bn')}$.
   For this, consider any $bk'$-set $K\subseteq[bn']$.
   We need to show that there is some $br'$-set $R\in\mathcal T^{(bn')}$ such that $R\subseteq K$.
   By construction of $\Pi$, there is some $\pi\in\Pi$ such that $|\pi^{-1}(j)\cap K| = k'$ for all $j=1,\dots,b$.
   Considering $K$ as an indicator vector, we look up $K_{\pi^{-1}(1)}, \dots, K_{\pi^{-1}(b)}$ in $\mathcal T^{(n')}$, which gives us $b$ disjoint $r'$-sets, $R'_1, \dots, R'_{b}$.
   By construction of $\mathcal T^{(bn')}$ there is a single $R\in\mathcal T^{(bn')}$ such that $R_{\pi^{-1}(j)} = R'_j$ for all $j$.
   Now, since $R'_j \subseteq K_{\pi^{-1}(j)}$ for all $j$, we get $R\subseteq K$ and so we have proven $\mathcal T^{(bn')}$ to be a correct $(bn', bk', br')$-system.

   Decoding $\mathcal T$ is fast, since we only have to do $|\Pi| \cdot b$ lookups in (enumerations of) $\mathcal T^{(n')}$.
   When evaluating $T^{(bn')}(K)$ we make sure we return every $br'$-set in $K$.
   Hence we return the entire ``product'' of unions:
   \begin{align*}
      T^{(bn')}(K) = \bigcup_{\pi\in\Pi}\{
            R_1\cup\dots\cup R_{b}
      : R_1\in T^{(n')}(K_{\pi^{-1}(1)}), R_2\in\dots\}
      .
   \end{align*}
  In total this takes time $|T^{(bn')}(K)|$ for the union product plus an overhead of
  $ |\Pi|b|\mathcal T^{(n')}| \le (en')^b r^{O(r'+b)} = r^{O(\sqrt{r})}  $ for the individual decodings.

   \paragraph{3.}
   %The next two steps can be done in either order, but doing hashing before partitioning gives a small performance improvement.
   Next we convert our $((k/a)^2, k/a, r)$ design, $\mathcal T^{(bn')}$ (note that $bn' = (k/a)^2$), to a $(n/a, k/a, r)$ design, call it $\mathcal T^{(n/a)}$.

   Let $\mathcal H$ be a perfect hash family of functions $h : [n/a] \to [(k/a)^2]$,
   such that for every $k/a$-set, $S \subseteq [n/a]$,
   there is an $h\in\mathcal  H$ such that $|h(S)| = k/a$.
   That is, no element of $S$ gets mapped to the same value.
   By lemma 3 in~\cite{alon2006algorithmic}, we can efficiently construct such a family of $(k/a)^4\log(n/a)$ functions.

   We will first describe the decoding function $T^{(n/a)} : \mathcal P([n/a]) \to \binom{[n/a]}{k/a}$, and then let $\mathcal T^{(n/a)} = T^{(n/a)}([n/a])$.
   For any set $S\subseteq[n/a]$ to be decoded,
   for each $h\in\mathcal  H$,
   we evaluate $T^{(bn')}(h(S))$ to get all $r$-sets $R\in \mathcal T^{(bn')}$ where  $R\subseteq h(S)$.
   For each such set, we return each set in
   \begin{align*}
      (h^{-1}(R_1)\cap S)
      \times (h^{-1}(R_2)\cap S)
      \times \dots \times (h^{-1}(R_r)\cap S),
   \end{align*}
   %where $R_i$ is the $i$th element of $R$ when considered as a $\zo^{bn'}$ vector.
   where $R_i$ is the $i$th element of $R$ when considered as a $[bn']^r$ vector.

   %In one line this means
   %$
   %T^{(n/a)}(S)=
   %\bigcup_{h\in H, R\in T^{(bn')}(h(S))}
   %\bigtimes_{1\le i\le r}
   %(h^{-1}(R_i)\cap S)
   %$.
   This takes time equal to the size of the above product (the number of results, $|T(S)|$) plus an overhead of $|\mathcal H|$ times the time to decode in $\mathcal T^{(bn')}$
   which is $|\mathcal H| r^{O(\sqrt r)} = e^\chi$ by the previous part.
   The other part of the decoding time, the actual size $|T^{bn'}(h(S))|$, is dominated by the size of the product.
%
   %\medskip
%
   %Something about the efficiency of the procedure.
   %It's pretty obvious that we just need to call $h(S)$ once (fast), then decode the sub-design (fast), back-project the $R$-set (fast) and calculate the product (time equal to results.)
   %One difference from the above calculation is that we can't cancel the factors of $\mathcal T^{(bn')}$ and $(|S|/(k/a)^2)^r$ as much, since $R$-sets that result in 0 generated sets will still have to be back projected.
   %Hence something like $|\mathcal T^{(bn')}|(1+(|S|/(k/a)^2)^r)$ may be more precice, in which case it may be advantageous to actually use $|T^{(bn')}(h(S))|$ rather than all of $\mathcal T^{(bn')}$.
%
   %\medskip
%
   To prove the the `efficient decoding' part of the lemma we thus have to show that the expected size of $T(S)$ is $\le (|S|a/k)^r e^\chi$
   for any $S\subseteq[n/a]$.
   (Note: this is for a set $S\subseteq[n/a]$, part four of the proof will extend to sets $S\subseteq[n]$ and that factor $a$ in the bound will disappear.)

   At this point we will add a random permutation, $\pi:[(k/a)^2]\to[(k/a)^2]$, to the preprocessing part of the lemma.
   This bit of randomness will allow us to consider the perfect hash-family as `black box' without worrying that it might conspire in a worst case way with our fixed family $\mathcal T^{(bn')}$.
   We apply this permutation to each function of $\mathcal H$, so we are actually returning
   \begin{align*}
      T^{(n/a)}(S)=
      \bigcup
      \left\{
         \begin{multlined}
      (h^{-1}(\pi^{-1} R_1)\cap S)
      \times (h^{-1}(\pi^{-1} R_2)\cap S)
      \times \dots \times (h^{-1}(\pi^{-1} R_r)\cap S)
      \\:
      \text{for all $R\in T^{(bn')}(\pi h(S))$ and $h\in\mathcal H$}.
   \end{multlined}
   \right\}
   \end{align*}
   %where we note that it doesn't matter if $\pi$ or subscript has higher the operator precedence $\pi R_1 = (\pi R)_1 = \pi(R_1)$.

   We can then show for any $S\subseteq[n/a]$:
\begin{align}
      E_\pi[|T^{(n/a)}(S)|]
      &=
      E_\pi\left[\sum_{h\in \mathcal H,\, R\in T^{(bn')}(\pi h(S))}
      \left|(h^{-1}(\pi^{-1} R_1)\cap S) \times \dots \times (h^{-1}(\pi^{-1} R_r)\cap S)\right|\right]
      \nonumber \\&=
      \sum_{h\in \mathcal H,\, R\in \mathcal T^{(bn')}}E_\pi\left[
            |(h^{-1}(\pi^{-1} R_1)\cap S)| \,\cdots\, |(h^{-1}(\pi^{-1} R_r)\cap S)|
      \cdot
      [R\subseteq \pi h(S)]\right]
      \nonumber \\&=
      |\mathcal T^{(bn')}|
      \sum_{h\in \mathcal H}E_{\pi}\left[
      |(h^{-1}(\pi^{-1} R_1)\cap S)| \,\cdots\, |(h^{-1}(\pi^{-1} R_r)\cap S)| \right]
      \label{eq:remove_subset}
      \\&\le
      |\mathcal T^{(bn')}|
      \sum_{h\in \mathcal H}E_{\pi}\left[|h^{-1}(\pi_1)\cap S|\right]^r
      \label{eq:use_maclaurin}
      \\&=
      |\mathcal T^{(bn')}|\, |\mathcal H|\, (|S|/(k/a)^2)^r
      %\\&=
      %(k/a)^4\log(n/a) (k/a)^r r^{O(b)} (|S|/(k/a)^2)^r
      \nonumber \\&=
      (|S|a/k)^r e^\chi
      .\nonumber
   \end{align}
   For~\eqref{eq:remove_subset} we used that
   \begin{align}
      [R\subseteq h(S)]
      = [\forall_{r\in R} r \in h(S)]
      = [\forall_{r\in R}h^{-1}(r)\cap S\neq\emptyset]
      \label{eq:subset_emptysets}
   \end{align}
   and so the value in the expectation was already 0 exactly when the Iversonian bracket was zero.

   For~\eqref{eq:use_maclaurin} we used the Maclaurin's Inequality~\cite{ben2014maclaurin} which says that $E(X_1X_2\dots X_r) \le (EX_1)^r$ when the $X_i$s are values sampled identically, uniformely at random without replacement from som set of non-negative values.
   In our case those values were sizes of partitions $h^{-1}(1)\cap S, \dots, h^{-1}(bn')\cap S$, which allowed us to bound the expectation as if $h$ had been a random function.

   \medskip

   We need to show that $T^{(n/a)}$ is a correct decoding function, that is
   $T^{(n/a)}(S) = \{R\in\mathcal T^{(n/a)} : R\subseteq S\}$,
   and the correctness of $\mathcal T^{(n/a)}$, that is $|S|\ge k/a $ implies $ T^{(n/a)}(S)\neq\emptyset$.

   For this, first notice that $T$ is monotone, that is
   if $S\subseteq U$ then $T(S) \subseteq T(U)$.
   This follows because $R\subseteq \pi h(S) \implies R\subseteq \pi h(U)$ and that each term $h^{-1}(R_i) \cap S$ is monotone in the size of $S$.
   This means we just need to show that $T(K)$ returns something for every $|K|=k$,
   since then $\mathcal T = T([n/a]) = T(\bigcup_K K) \supseteq \bigcup T(K)$ will return all these things.

   Hence, consider a $k$-set, $K\subseteq [n/a]$.
   By the property of $\mathcal H$, there must be some $h\in \mathcal \pi H$ such that $|h(K)|=k$,
   and by correctness of $\mathcal T^{(bn')}$ we know there is some $r$-set, $R\in T^{(bn')}(h(K))$.
   Now, for these $h$ and $R$, since $R \subseteq h(K)$ and using~\eqref{eq:subset_emptysets} we have that $(h^{-1}(R_1)\cap K)\times\dots$ is non-empty, which is what we wanted.
   Conversely, consider some $R\in \mathcal T^{(n/a)}=T^{(n/a)}([n/a])$ such that $R\subseteq K$, then
   $R\in h^{-1}(R'_1) \times h^{-1}(R'_2) \dots$ for some $R'\in\mathcal T^{(bn')}$ and $h(R)\subseteq h(K)$.
   However $h(R)$ is exactly $R'$,
   since $R_i \in h^{-1}(R'_i)\implies h(R_i) = R'_i$,
   which shows that $T^{(n/a)}(K)$ returns all the set we want.

   \paragraph{4.}

   Finally we convert our $(n/a, k/a, r)$ design, $\mathcal T^{(n/a)}$ to an $(n, k, r)$ design, call it $\mathcal T$.
   We do this by choosing a random permutation $\pi:[n]\to[n]$ and given any set $S\subseteq[n]$ we decode it as
   \begin{align*}
      T(S)
      = T^{(n/a)}(\pi S \cap \{1,\dots,n/a\})
      \cup \dots \cup
      T^{(n/a)}(\pi S \cap \{n-n/a+1,\dots,n\})
      . 
   \end{align*}

   To see that this is indded an $(n, k, r)$ design, consider any set $K\subseteq[n]$ of size $|K|=k$, there must be one partition $K\cap\{1\dots,n/a\}, \dots$ that has at least the average size $k/a$.
   Since $\mathcal T^{(n/a)}$ is a $(n/a,k/a,r)$ design, it will contain a set $R\subseteq K\cap\{in-n/a+1,\dots,in\} \subseteq K$ which we return.

   It remains to analyze the size of $T(S)$, which may of course get large if we are so unlucky that $\pi$ places much more than the expected number of elements in one of the partitions.
   In expectation this turns out to not be a problem, as we can see as follows:
   \begin{align*}
      E_\pi|T(S)|
      &= \sum_i E_\pi\left[\left|T^{(n/a)}(\pi S \cap p_i)\right|\right]
      \\&= a \sum_{s} E\left[\left|T^{(n/a)}(\pi S \cap p_1)\right| \,\big|\, |\pi S\cap p_1|=s\right] \, \Pr[|\pi S\cap p_1| = s]
      \\&=
      a \sum_{s} (sa/k)^r e^\chi \, 
      \binom{|S|}{s}\binom{n-|S|}{n/a-s}\bigg/\binom{n}{n/a}
      \\&\le
      e^\chi \sum_{s}
      \frac{\binom{s}{r}}{\binom{k/a}{r}}
      \frac{\binom{|S|}{s}\binom{n-|S|}{n/a-s}}{\binom{n}{n/a}}
      \\&=
      e^\chi
      \frac{\binom{|S|}{r}}{\binom{k/a}{r}\binom{n}{n/a}}
      \sum_{s} \binom{|S|-r}{s-r}\binom{n-|S|}{n/a-s}
      \\&=
      e^\chi
      \frac{\binom{|S|}{r}\binom{n-r}{n/a-r}}{\binom{k/a}{r}\binom{n}{n/a}}
      =
      e^\chi
      \frac{\binom{|S|}{r}\binom{n/a}{r}}{\binom{k/a}{r}\binom{n}{r}}
      \\&\le
      e^\chi
      (|S|a/k)^r e^{r^2/(k/a)} a^{-r}
      =
      (|S|/k)^r e^\chi.
   \end{align*}
   Here we used Vandermonde convolution to complete the sum over $s$, and then eventually lemma~\ref{lem:ratio} to bound the binomial ratios.
   This completes the proof of lemma~\ref{lem:turan}.
\end{proof}

\subsection{Integrality considerations}

In the proof, we needed the following quantities to be integral:
$b = r/b = \sqrt{r}$,
$a=k r^{-3/2} \log(r^{3/2})$,
$k^2/(a^2 b) = k/a = r^{3/2}/\log(r^{3/2})$,
$k/(ab) = r/\log(r^{3/2})$,
$n/a$.

It suffices to have
$\sqrt{r}$ and $r/\log(r^{3/2})$ integral,
and that the later divides $k$.

It is obviously ridiculous to require that $r$ is a square number.
Or is it?
You can actually make a number square by just changing it by a factor $1+2/\sqrt{r}$.
That would only end up giving us an $e^{O(\sqrt{r})}$, so maybe not so bad?

To make $r/\log(r^{3/2})$ integral, we can multiply with a constant.
Since it didn't matter that we divided by a $\log$, surely it doesn't matter that we multiply with a constant.

To make $r/\log(r^{3/2})$ divide $k$, we need $k$ to have some divisors.
We can't just round $k$ to say, a power of two, since that could potentially change it by a constant factor, which would come out of $(n/k)^r$.
We can change $k$ with at most $1+1/\sqrt{r}$. So $1+1/\sqrt{k}$ would be just fine.
Of course we can change it by an additive $r/\log(r^{3/2})$. That corresponds to a factor about $1+r/k$.
Since $k > r^{3/2}$ that is fine!
Or maybe we'll subtract that, because then it is still a valid $(n,k,r)$ design.
In the same way, if we round $r$ up to the nearest square root, we don't have to make the changes in the later calculations, but they can be kept intirely inside the lemma.

\section{Conclusion and Open Problems}

We have seen that, perhaps surprisingly, there exists a relatively general way of creating efficient Las Vegas versions of state-of-the art high-dimensional search data structures.

As bi-products we found efficient, explicit constructions of large Turán systems and covering codes for pairs.
We also showed an efficient way to do dimensionality reduction in hamming space without false negatives.

The work leaves a number of open questions for further research:
\begin{enumerate}[\hspace{0cm}1)]
   \item
      Can we make a completely deterministic high-dimensional data structure for the proposed problems?
      Cutting the number of random bits used for Las Vegas guarantees would likewise be interesting.
      The presented algorithms both use $O(d\log d)$ bits, but perhaps limited independence could be used to show that $O(\log d)$ suffice?
   \item
      Does there exist Las Vegas data structures with performance matching that of data-dependent LSH data structures?
      This might connect to the previous question, since a completely deterministic data structure would likely have to be data-dependent.
      However the most direct approach would be to repeat~\cite{andoni2017optimal}, but find Las Vegas constructions for the remaining uses of Monte Carlo randomness, such as clustering.
   \item
      By reductions, our results extend to $\ell_2$ and $\ell_1$ with exponent $n^{1/c}$.
      This is optimal for $\ell_1$, but for $\ell_2$ we would hope to get $n^{1/c^2}$.
      Can our techniques be applied to yield such a data structure?
      Are there other interesting LSH algorithms that can be made Las Vegas using our techniques?
      The author conjectures that a space/time trade-off version of the presented algorithm should follow easily following the approach of~\cite{andoni2017optimal,laarhoven2015tradeoffs, christiani2016framework}.
   \item
      In the most general version, we we get the overhead term $(\log n)^{-1/4}$ in our $\rho$ value.
      Some previous known LSH data structures also had large terms, such as \cite{andoni2006near}, which had a $(\log n)^{-1/3}$ term and~\cite{andoni2017optimal}, which has $(\log\log n)^{-\Theta(1)}$,
      but in general most algorithms have at most a $(\log n)^{-1/2}$ term.

      Can we improve the overhead of the approach given in this paper?
      Alternatively, is there a completely different approach, that has a smaller overhead?
\end{enumerate}

\subsection{Acknowledgements}
The author would like to thank Rasmus Pagh, Tobias Christiani and the members of the Scalable Similarity Search project for many rewarding discussions on derandomization and set similarity data structures.
Further thanks to Omri Weinstein, Rasmus Pagh, Martin Aumüller, Johan Sivertsen and the anonymous reviewers for useful comments on the manuscript;
and to the people at University of Texas, Austin, for hosting me while doing much of this work.
An extra thanks to Piotr Wygocki for pointing out the need for a deterministic reduction from $\ell_1$ to Hamming space.

The research leading to these results has received funding from the European Research Council under the European Union’s 7th Framework Programme (FP7/2007-2013) / ERC grant agreement no. 614331.

\bibliographystyle{apalike}
\balance
\bibliography{lasvegas}

\section{Appendix}

\subsection{Explicit reduction from $\ell_1$ to Hamming}
\label{ell1tohamming}

\begin{theorem}
   For $d, r, c\ge1$ and a set of points $P\subseteq \mathbb R^d$ of size $|P|=n$,
   there is a function $f :\mathbb  R^d \to \{0,1\}^{b}$ where $b=2d^2\epsilon^{-3}\log n$, such that
   for any two points $x\in\mathbb  R^d, y\in P$,
%   \begin{align*}
%      \left| \frac{2dcr}{b} \frac{\|f(x)-f(y)\|}{\min(\|x-y\|_1,cr)} -1 \right| \le
%      \epsilon
%   \end{align*}
   \begin{enumerate}
      \item if $\|x-y\|_1 \le r$ then $\|f(x)-f(y)\|_1 \le (1+\epsilon)S$,
      \item if $\|x-y\|_1 \ge cr$ then $\|f(x)-f(y)\|_1 \ge (1-\epsilon)cS$,
   \end{enumerate}
   and the scale factor $S = b/(2dcr) = (d\log n)/(cr\epsilon^3)$.
\end{theorem}
\begin{proof}
   First notice that we can assume all coordinates are at most $rn$.
   This can be assured by imposing a grid of side length $2rn$ over the points of $P$, such that no point is closer than distance $r$ from a cell boundary.
   Since points $x,y\in\mathbb R^d$ in different cells must now be more than distance $r$ from each other, we can set the embedded distance to $cS$ by prefixing points from different cells with a different code word.
   The grid can be easily found in time $O(dn)$ by sweeping over each dimension seperately.

   %TODO: How much does this cost in the dimension?

   Notice that for actual data structure purposes, we can even just process each cell seperately and don't have to worry about separating them.

   \medskip

   By splitting up each coordinate into positive and negative parts, we can further assume each coordinate of each vector is positive.
   This costs a factor of 2 in $d$.

   Next, if we have an $2\epsilon r/d$ grid, then there is always a grid point within $\ell_1$-distance $\epsilon r$.
   That means if we multiply each coordinate by $d/(2\epsilon r)$ and round the coordinates to nearest integer, we get distances are changed by at most $\epsilon r$.

   We are now ready for the main trick of the reduction.
   Let $M$ be the largest coordinate, which we can assume is at most $dn/\epsilon$, and $R=dc/(2\epsilon)$ be the value of $cr$ after scaling and rounding.
   For each coordinate we map $[M] \to [\lfloor M/R\rfloor]^R$ by
   $h(c) := \langle
   \lfloor\frac{x}{R}\rfloor,
   \lfloor\frac{x+1}{R}\rfloor,
   \dots,
   \lfloor\frac{x+R-1}{R}\rfloor\rangle$.
   The point of this mapping is that for every $c_1,c_2\in[M]$, $\dist(h(c_1),h(c_2))=\min(|c_1-c_2|, R)$, where $\dist$ is hamming distance.

   \begin{figure}[h]
      \centering
      \begin{tikzpicture}
         \foreach \i in {0,...,7} {
            \draw[thick] (\i+1,-1) -- (\i+1,2);
         }
         \draw[thick] (-1,1) -- (8,1);
         \draw[thick] (-1,0) -- (8,0);
         \draw[thick, double] (0,-1) -- (0,2);

         \node at (-.5,1.5) {$100$};
         \foreach \i in {0,...,7} {
            \node at (\i+.5,1.5) {$ \pgfmathparse{ int((100+\i)/8) }\pgfmathprintnumber\pgfmathresult $};
            \draw[thick] (\i+1,1) -- (\i+1,2);
         }
         \node at (-.5,.5) {$105$};
         \foreach \i in {0,...,7} {
            \node at (\i+.5,.5) {$ \pgfmathparse{ int((105+\i)/8) }\pgfmathprintnumber\pgfmathresult $};
            \draw[thick] (\i+1,0) -- (\i+1,1);
         }
         \foreach \i in {0,...,3} {
            \node at (\i+.5,-.5) {$*$};
         }
         \node at (7.5,-.5) {$*$};
      \end{tikzpicture}
      \caption{Mapping $100$ and $105$ to $[\lfloor100/8\rfloor]^8$ preserving $\ell_1$ distance in Hamming distance.}
   \end{figure}

   All that's left is to use a code with good minimum and maximum distance to map down into $\{0,1\}$.
   A random code with bit length $k=4\epsilon^{-2}(\log 4n)$ suffices.
   To see this, let $X$ be a binomial random variable, $X\sim B(k,1/2)$.
   Then
   \begin{align*}
      \Pr[(1-\epsilon)k/2\le C \le(1+\epsilon)k/2] \le 2e^{-\epsilon^2 k/2} \le 1/(8n^2)
   \end{align*}
   so by union bound over all $\binom{M/R}{2}\le 2n^2$ pairs of values, we have constant probability that the code works.
   For a given code, we can check this property deterministically in time $kn^2$, so we can use rejection sampling and generate the code in time $\approx O(n^2)$.
   Of course, $n^2$ time may be too much.
   Luckily there are also explicit codes with the property, such as those by Naor and Naor~\cite{naor1993small}.

   The complete construction follows by concatenating the result of $h$ on all coordinates.
\end{proof}

See~\cite{indyk2007uncertainty} for an explicit reduction from $\ell_2$ to $\ell_1$.

\subsection{The Ratio of Two Binomial Coefficients}

Classical bounds for the binomial coefficient:
$(n/k)^k \le {n\choose k} \le (en/k)^k$
give us simple bounds for binomial ratios, when $n\ge m$:
$(n/em)^k \le {n\choose k}\big/{m\choose k} \le (en/m)^k$.
The factor $e$ on both sides can often be a nuisance.

Luckily tighter analysis show, that they can nearly always be either removed or reduced.
Using the fact that $\frac{n-i}{m-i}$ is increasing in $i$ for $n\ge m$, we can show
   $ {n\choose k}\big/{m\choose k} = \prod_{i=0}^{k-1}\frac{n-i}{m-i} \ge \prod_{i=0}^{k-1}\frac{n}{m} =\left(\frac nm\right)^k $.
This is often sharp enough, but on the upper bound side, we need to work harder to get results.

Let $\HE(x) = x\log1/x + (1-x)\log1/(1-x)$ be the binary entropy function,
\begin{lemma}\label{lem:ratio}
   For $n\ge m\ge k\ge 0$ we have the following bounds:
   \begin{align*}
      \left(\frac nm\right)^k
      \le \left(\frac nm\right)^k \exp\left(\frac{n-m}{nm}\frac{k(k-1)}2\right)
      \le {n\choose k}\bigg/{m\choose k}
      \le \frac{\exp\left(n\HE(k/n)\right)}{\exp\left(m\HE(k/m)\right)}
      \le \left(\frac{n}{m}\right)^ke^{k^2/m}
   \end{align*}
\end{lemma}
If $m \ge n$ we can simply flip the inequalities and swap $n$ for $m$.
Note that $(n/em)^k\le(n/m)^k$ and $e^{k^2/m}\le e^k$, so the bounds strictly tighten the simple bounds states above.

Especially the entropy bound is quite sharp, since we can also show:
$
   {n\choose k}\big/{m\choose k}
   \ge
   \frac{\exp\left((n+1)\HE(k/(n+1))\right)}{\exp\left((m+1)\HE(k/(m+1))\right)}
$,
though for very small values of $k$, the lower bound in the theorem is actually even better.
We can also get a feeling for the sharpness of the bounds, by considering the series expansion of the entropy bound at $k/m\to 0$:
$
   \frac{\exp\left(n\HE(k/n)\right)}{\exp\left(m\HE(k/m)\right)}
   =
   \left(\frac nm\right)^k \exp(\frac{n-m}{nm}\frac{k^2}2 + O(k^3/m^2))
$.

For the proofs, we'll use some inequalities on the logarithmic function from~\cite{topsok2006some}:
\begin{align}
   \log (1+x) &\ge x/(1+x) \label{eq:log-lower-1}\\
   \log (1+x) &\ge 2 x/(2 + x) \text{ for } x\ge0 \label{eq:log-lower-2}\\
   \log (1+x) &\le x (2 + x)/(2 + 2x) \text{ for } x\ge 0\label{eq:log-upper-1}
   .
\end{align}
In particular \eqref{eq:log-lower-2} and \eqref{eq:log-upper-1} imply the following bounds for the entropy function:
\begin{align}
   \HE(x) &\le x\log1/x + x(2-x)/2 \label{eq:entro-upper}\\
   \HE(x) &\ge x\log1/x + 2x(1-x)/(2-x) \label{eq:entro-lower}
   ,
\end{align}
which are quite good for small $x$.

We'll prove theorem~\ref{lem:ratio} one inequality at a time, starting from the left most:
\begin{proof}
   The first inequality follows simply from $\frac{n-m}{nm}\frac{k(k-1)}2\ge0$, which is clear from the conditions on $n\ge m\ge k$.

   The second inequality we prove by using~\eqref{eq:log-lower-1}, which implies $1+x\ge \exp(x/(1+x))$, to turn the product into a sum:
   \begin{align*}
      {n\choose k}\bigg/{m\choose k}
      &=
      \prod_{i=0}^{k-1}\frac{n-i}{m-i}
      \\&=
      \left(\frac nm\right)^k
      \prod_{i=0}^{k-1}\frac{1-i/n}{1-i/m}
      \\&=
      \left(\frac nm\right)^k
      \prod_{i=0}^{k-1}\left(1+\frac{i/m-i/n}{1-i/m}\right)
      \\&\ge
      \left(\frac nm\right)^k
      \exp\left(\sum_{i=0}^{k-1}\frac{i(n-m)}{(n-i)m}\right)
      \\&\ge
      \left(\frac nm\right)^k
      \exp\left(\sum_{i=0}^{k-1}i\frac{n-m}{nm}\right)
      \\&=
      \left(\frac nm\right)^k
         \exp\left(\frac{k(k-1)}2\frac{n-m}{nm}\right).
   \end{align*}

   For the entropy upper bound we will use an integration bound, integrating $\log(n-i)/(m-i)$ by parts:
   \begin{align*}
      {n\choose k}\bigg/{m\choose k}
      &=
      \prod_{i=0}^{k-1}\frac{n-i}{m-i}
      \\&=
      \exp\left(
         \sum_{i=0}^{k-1}\log\frac{n-i}{m-i}
      \right)
      \\&\le
      \exp\left(
         \int_{0}^{k}1\log\frac{n-x}{m-x} dx
      \right)
      \\&=
      \exp\left(x \log\frac{n-x}{m-x}\bigg|_0^k
      - \int_0^k x \left(\frac{1}{m-x}-\frac{1}{n-x} \right)dx \right)
      \\&=
      \exp\left(k \log\frac{n-k}{m-k}
      + \int_0^k \left(\frac{m}{m-x}-\frac{n}{n-x}\right) dx \right)
      \\&=
      \exp\left(k \log\frac{n-k}{m-k}
      - \bigg| m\log\frac{1}{m-x}-n\log\frac1{n-x}\bigg|_0^k \right)
      \\&=
      \exp\left(n\HE(k/n)-m\HE(k/m)\right).
   \end{align*}
   The integral bound holds because $\log\frac{n-i}{m-i}$ is increasing in $i$, and so $\log\frac{n-i}{m-i}\le\int_i^{i+1}\log\frac{n-x}{m-x}dx$.
   We see that $\frac{n-i}{m-i}$ is increasing by observing $\frac{n-i}{m-i}=\frac{n}{m}+\frac{in/m-i}{m-i}$ where the numerator and denominator of the last fraction are both positive.
   The entropy lower bound, mentioned in the discussion after the theorem, follows similarly from integration, using $\log\frac{n-i}{m-i}\ge\int_{i-1}^{i}\log\frac{n-x}{m-x}dx$.

   For the final upper bound, we use the bounds~\eqref{eq:entro-upper} and~\eqref{eq:entro-lower} on $\HE(k/n)$ and $\HE(k/m)$ respectively:
   \begin{align*}
      \frac{\exp\left(n\HE(k/n)\right)}{\exp\left(m\HE(k/m)\right)}
      \le \left(\frac nm\right)^k \exp\left(\frac{k^2}{2}\left(\frac1{m-k/2}-\frac1n\right)\right)
      \le \left(\frac nm\right)^k \exp\left(\frac{k^2}{m}\right)
      .
   \end{align*}
\end{proof}

\subsection{Proof of lemma~\ref{lem:int}}
Recall the lemma:
\intt*
\begin{proof}
\begin{figure}
   \centering
   \begin{tikzpicture}
      \draw [thick,decorate,decoration={brace,amplitude=10}]
         (0,3.05) -- (5,3.05) node[midway,anchor=south,yshift=10] {$d-r$};
      \draw [thick,decorate,decoration={brace,amplitude=10}]
         (5,3.05) -- (9,3.05) node[midway,anchor=south,yshift=10] {$r$};

      \draw[thick] (0,0) -- (9,0) -- (9,3) -- (0,3) -- (0,0);
      \draw[thick] (0,1) -- (9,1);
      \draw[thick] (0,2) -- (9,2);
      \draw[thick] (5,0) -- (5,3);
      \node at (2.5,2.5) {0};
      \node at (2.5,1.5) {0};
      \node at (2.5,.5) {$j$};
      \node at (7,2.5) {0};
      \node at (7,1.5) {1};
      \node at (7,.5) {$i$};
      \node[anchor=east] at (0,2.5) {$x$};
      \node[anchor=east] at (0,1.5) {$y$};
      \node[anchor=east] at (0,.5) {$z$};
   \end{tikzpicture}
   \caption{To calculate how many points are within distance $t$ from two points $x$ and $y$, we consider without loss of generality $x=0\dots0$.
   For a point, $z$, lying in the desired region, we let $i$ specify the number of $1$'s where $x$ and $y$ differ, and $j$ the number of $1$'s where they are equal.
   With this notation we get $d(x,z)=i+j$ and $d(y,z)=j+r-i$.
   \label{fig:iandj}
   }
\end{figure}

\begin{figure}
   \centering
   \begin{tikzpicture}
      \node[anchor=south west, inner sep=0pt] (background) at (0,0) {\includegraphics[width=.8\textwidth]{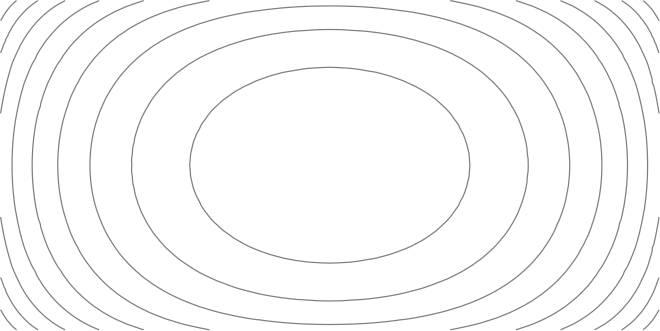}};
      \begin{scope}[x={(background.south east)},y={(background.north west)}]
         % pentagon
         \draw[thick, black] (0,0) -- (.15,0) -- (.4,.5) -- (.15,1) -- (0,1);
         % labels
         \node[anchor=east] at (.3,.3) {$j-i=t-r$};
         \node[anchor=east] at (.3,.7) {$j+i=t$};
         % i-axis
         \node[anchor=east] at (-.1,.5) {$i$};
         \draw[thick, black] (0,0) -- (0,1);
         \draw[black, dashed] (0,.5) -- (1,.5);
         \node[anchor=east] at (0,1) {$r$};
         \node[anchor=east] at (0,.5) {$\frac{r}2$};
         \node[anchor=east] at (0,0) {$0$};
         % j-axis
         \node[anchor=north] at (.5,-.1) {$j$};
         \draw[thick, black] (0,0) -- (1,0);
         \draw[black, dashed] (.5,0) -- (.5,1);
         \draw[black, dashed] (.4,.5) -- (.4,0);
         \node[anchor=north] at (0,0) {$0$};
         \node[anchor=north] at (.15,0) {$t-r$};
         \node[anchor=north] at (.4,0) {$t-\frac{r}2$};
         \node[anchor=north] at (.5,0) {$\frac{d-r}2$};
         \node[anchor=north] at (1,0) {$d-r$};
      \end{scope}
   \end{tikzpicture}
   \caption{
      A contour plot over the two dimensional binomial.
      The pentagon on the left marks the region over which we want to sum.
      For the upper bound we sum $i$ from 0 to $r$ and $j$ from 0 to $t-r/2$.
   }
   \label{fig:regions}
\end{figure}

      From figure~\ref{fig:iandj} we have that
      $I = \sum_{\substack{i+j\le t\\ j-i\le t-r}}{r\choose i}{d-r\choose j}$,
      and from monotonicity (and figure~\ref{fig:regions}) it is clear that
      $
         \binom{r}{r/2} \binom{d-r}{t-r/2}
         \le I \le
         \sum_{\substack{0\le i\le r\\0\le j\le t-r/2}} {r\choose i} \binom{d-r}{j}
      $.

      We expand the binomials using Stirling's approximation:
      %\begin{align*}
      $
         \frac{\exp(n H(k/n))}{\sqrt{8(1-k/n)k}}\le
         \binom{n}{k} \le \sum_{i\le k}\binom{n}{i}\le \exp(n H(k/n))
      $
      %\end{align*}
      where $H(x) = x\log\frac1x + (1-x)\log\frac1{1-x}$ is the binary entropy function, which we bound as
      %\begin{align*}
      $
         \log2 - 2 (\tfrac12-x)^2 - 4(\tfrac12-x)^4
         \le H(x) \le \log2 - 2 (\tfrac12 - x)^2
      $.
      %\end{align*}
      We then have for the upper bound:
      \begin{align*}
         I 2^{-d}
         \le 2^{r-d} \exp\left[(d-r)H\left(\tfrac{t-r/2}{d-r}\right)\right]
         \le \exp\left[-\tfrac{s^2}{2(1-r/d)}\right]
      \end{align*}
      And for the lower bound:
      \begin{align*}
         I 2^{-d}
         \ge 2^{-d} \binom{r}{r/2} \binom{d-r}{t-r/2}
         &\ge 
            \frac{2^{r-d}}{\sqrt{2r}}
            \frac{\exp\left[(d-r)H\left(\tfrac{t-r/2}{d-r}-\log 2\right)\right]}{\sqrt{8(1-\tfrac{t-r/2}{d-r})(t-r/2)}}\\
         &\ge
            \exp[-\tfrac{s^2}{2(1-r/d)}]
            \frac{\exp[-\tfrac{s^4}{4(1-r/d)^3d}]}
            {\sqrt{4r(d-r)(1-\tfrac{d s^2}{(d-r)^2})}}\\
         &\ge
            \exp[-\tfrac{s^2}{2(1-r/d)}]
            \frac1d
            \frac{1-2s^4/d}{\sqrt{1-4s^2/d}}
            ,
      \end{align*}
      where for the last inequality we used the bound $e^x \ge 1+x$.
      The last factor is monotone in $s$ and we see that for $s\le d^{1/4}/2$ it is $\ge\frac78\left(1-1/\sqrt{d}\right)^{-1/2}\ge\frac78$, which gives the theorem.
   \end{proof}
   The factor of $1/d$ can be sharpened a bit, e.g. by using the two dimensional Berry-Essen theorem from~\cite{bentkus2005lyapunov}.
\end{document}